\documentclass[11pt]{article}

\usepackage{lmodern}
\usepackage[T1]{fontenc}
\usepackage[utf8]{inputenc}
\usepackage{fullpage}
\usepackage{hyperref}
\usepackage{amsmath, amssymb, amsthm}
\usepackage{mathtools}
\usepackage{verbatim}
\usepackage{multicol}
\usepackage{authblk}
\usepackage[svgnames]{xcolor}
\usepackage{tikz}
\usetikzlibrary{math, calc}

\definecolor{DarkGreen}{rgb}{0.1,0.5,0.1}
\definecolor{DarkRed}{rgb}{0.5,0.1,0.1}
\definecolor{DarkBlue}{rgb}{0.1,0.1,0.5}
\usepackage[capitalise,nameinlink]{cleveref}
\definecolor{links}{RGB}{11, 85, 255}
\definecolor{cites}{RGB}{0, 200, 0}
\definecolor{urls}{RGB}{255, 116, 0}
\hypersetup{colorlinks={true},linkcolor={links},citecolor=[named]{green},urlcolor={blue}}

\usepackage{thmtools, thm-restate}
\usepackage[numbers,sort&compress]{natbib}

\usepackage[mathscr]{euscript}
 \let\mathscr\relax% just so we can load this and rsfs
\usepackage[scr]{rsfso}

\def\epsilon{\varepsilon}

\renewcommand{\hat}{\widehat}

\DeclareMathOperator*{\argmin}{\mathrm{argmin}}

\newtheorem{theorem}{Theorem}
\newtheorem{proposition}{Proposition}
\newtheorem{definition}{Definition}

\newtheorem{remark}{Remark}
\newtheorem{lemma}{Lemma}

\newcommand{\vecc}[1]{\ensuremath{\mathbf{#1}}}
\newcommand{\mech}{\ensuremath{\mathcal{M}}}

\newcommand{\budget}{B}
\newcommand{\timeLimit}{T}

\usepackage{algorithm}
\usepackage{algpseudocode}
\renewcommand{\algorithmiccomment}[1]{\bgroup\hfill\small\textcolor{gray}{//~#1}\egroup}
\algnotext{EndFor}
\algnotext{EndIf}
\algnotext{EndWhile}
\algnotext{EndProcedure}
\makeatletter
\renewcommand{\ALG@name}{Procedure}
\makeatother

\title{Parallel Contests for Crowdsourcing Reviews: Existence and Quality of Equilibria\thanks{\, A conference version appears in the Proceedings of the 4th ACM Conference on
Advances in Financial Technologies, AFT 2022. This work was partially supported by the ERC Advanced Grant 788893 AMDROMA ``Algorithmic and Mechanism Design Research in Online Markets'', and the MIUR PRIN project ALGADIMAR ``Algorithms, Games, and Digital Markets''.}}

\author[1]{Georgios Birmpas}
\author[2, 4]{Lyudmila Kovalchuk}
\author[5]{Philip Lazos}
\author[3, 4]{Roman Oliynykov}

\affil[1]{University of Liverpool}
\affil[ ]{{\small\textsf{\href{mailto:G.Birmpas@liverpool.ac.uk}{G.Birmpas@liverpool.ac.uk}}}\smallskip}

\affil[2]{{National Technical University of Ukraine ``Igor Sikorsky Kyiv Polytechnic Institute''}\smallskip}
\affil[3]{{V.N.Karazin Kharkiv National University}\smallskip}
\affil[4]{{IOHK}\smallskip}		
\affil[ ]{
{\small
    \textsf{\{\href{mailto:lyudmila.kovalchuk@iohk.io}{lyudmila.kovalchuk},
    \href{mailto:roman.oliynykov@iohk.io}{roman.oliynykov}\}@iohk.io}}}

   \affil[5]{Jump Trading}		
\affil[ ]{
{\small
    \textsf{
    \href{mailto:plazos@gmail.com}{plazos@gmail.com}}}} 

\date{\today}

\begin{document}

\maketitle

%%
%% The abstract is a short summary of the work to be presented in the
%% article.
\begin{abstract}
  Part of the design of many blockchains and cryptocurrencies includes a \emph{treasury}, which periodically allocates collected funds to various projects that could be beneficial to their ecosystem. These projects are then voted on and selected by the users of the respective cryptocurrency. To better inform the users' choices, the proposals can be reviewed, in distributed fashion.
  Motivated by these intricacies, we study the problem of crowdsourcing reviews for different proposals, in parallel. During the reviewing phase, every reviewer can select the proposals to write reviews for, as well as the quality of each review. The quality levels follow certain very coarse community guidelines (since the review of the reviews has to be robust enough, even though it is also crowdsourced) and can have values such as `excellent' or `good'. Based on these scores and the distribution of reviews, every reviewer will receive some reward for their efforts. In this paper, we consider a simple and intuitive reward scheme and show that it always has Nash equilibria, under two different scenarios. In addition, we show that these equilibria guarantee constant factor approximations for two natural metrics: the total quality of all reviews, as well as the fraction of proposals that received at least one review, compared to the optimal outcome. %\textcolor{red}{This is a bit inaccurate in the sense that for the second objective we have an increased budget, rephrase?}.
\end{abstract}

%%
%% The code below is generated by the tool at http://dl.acm.org/ccs.cfm.
%% Please copy and paste the code instead of the example below.
%%
% \begin{CCSXML}
% <ccs2012>
%  <concept>
%   <concept_id>10010520.10010553.10010562</concept_id>
%   <concept_desc>Computer systems organization~Embedded systems</concept_desc>
%   <concept_significance>500</concept_significance>
%  </concept>
%  <concept>
%   <concept_id>10010520.10010575.10010755</concept_id>
%   <concept_desc>Computer systems organization~Redundancy</concept_desc>
%   <concept_significance>300</concept_significance>
%  </concept>
%  <concept>
%   <concept_id>10010520.10010553.10010554</concept_id>
%   <concept_desc>Computer systems organization~Robotics</concept_desc>
%   <concept_significance>100</concept_significance>
%  </concept>
%  <concept>
%   <concept_id>10003033.10003083.10003095</concept_id>
%   <concept_desc>Networks~Network reliability</concept_desc>
%   <concept_significance>100</concept_significance>
%  </concept>
% </ccs2012>
% \end{CCSXML}

% \ccsdesc[500]{Computer systems organization~Embedded systems}
% \ccsdesc[300]{Computer systems organization~Redundancy}
% \ccsdesc{Computer systems organization~Robotics}
% \ccsdesc[100]{Networks~Network reliability}

%%
%% Keywords. The author(s) should pick words that accurately describe
%% the work being presented. Separate the keywords with commas.
%\keywords{blockchain, equilibria, contests, price of anarchy}

%%
%% This command processes the author and affiliation and title
%% information and builds the first part of the formatted document.

%%%%%%%%%%%%%%%%%%%%%%%%%%%
\section{Introduction}
%%%%%%%%%%%%%%%%%%%%%%%%%%%

Since the invention of Bitcoin \cite{nakamoto2008bitcoin} in 2009, cryptocurrencies and other blockchain platforms have enjoyed a massive increase in popularity and market capitalization. Their development came with a promise for decentralization; instead of having a selected group of people manage their current operation and future direction, the users themselves should eventually express their preferences and guide every decision. This was first achieved just for the \emph{consensus layer}. More recent platforms such as Tezos and Polkadot have implemented forms of \emph{on-chain} governance. In addition to having a public discourse about their development on GitHub or in internet forums, there are formal governance processes and voting procedures to elect councils, hold referenda and adopt new changes into the protocol's codebase. The next piece missing in the puzzle of decentralization is how to \emph{fund} the necessary new features, with the final decision coming directly from aggregating user preferences. Some blockchains such as Cardano and Dash take this approach, having a publicly controlled treasury which allocates funds to proposals generated by the community. This process requires the careful design of many mechanisms, drawing from areas such as Crowdsourcing, Participatory Budgeting \cite{benade2021preference} and Distortion \cite{AnshelevichF0V21}. In this work, we focus on one important part of this process: how to provide the right incentives for the community to produce high quality \emph{reviews} across many proposals.

We provide a high level description of the relevant modules of Project Catalyst, the mechanism used by Cardano's treasury, as motivation for our modelling, assumptions, and results. A new funding round takes place every 12 weeks. The latest (ongoing as of May 2022) \emph{Fund8} allocated about $\$16,000,000$ worth of ADA (Cardano's native currency), with $5\%$ of the total given as rewards to reviewers. After the members of the community submit their proposals, there are 2 kinds of entities involved in the reviewing process: veteran
community advisors (vCA's) and community advisors (CA's). The CA's start by writing reviews for proposals. The reviews contain a written evaluation of the proposal, as well as a numerical score. The reviews (following some filtering for profanity, plagiarism, etc.) are published immediately under the CA's pseudonym, while the true identity of the person is usually secret. Then, the vCA's write meta-reviews (i.e., they review the reviews submitted by the CA's). This evaluation is only on the quality of the review itself, not its score or subjective opinion. The evaluation follows specific, well-known guidelines and is \emph{very coarse}: reviews can either be `excellent', `good' or `filtered out' (which are discarded). The CA's are then rewarded based on the quality of their reviews, as well as their choice over proposals to review, relative to actions of the other CA's.

The restrictions faced by Project Catalyst would apply to most blockchain solutions. Due to the crowdsourced nature of the process and their anonymous nature any mechanism used needs to be simple, clearly fair, and robust. Contrary to the approach possible in budget feasible mechanism design, or the use of the revelation principle in the contest literature, it is very difficult to elicit truthful information from the voters: there is no way to charge them for not delivering on their promise or to finely differentiate their quality of service. Additionally, requiring them to `lock' a certain amount of funds prior to producing the reviews would be highly unpopular as well, in a platform whose goal is to maximize participation. Essentially, in order to maximize robustness (to agents that might promise to participate and then disappear, or who could flood the system with useless reviews, etc.) rewards can only be given for complete reviews, after these have been produced. In that regard, we need to be careful so that even users who might not be particularly skilled have a reason to participate and build experience for future iterations. Finally, the system should be designed without any assumption about the CA's: given their anonymity, only they should have the freedom to select which proposals to write reviews for, given their motivation, expertise and time restrictions. As we will show, even though we shift the entire weight of the system to the anonymous users, good quality outcomes (comparable to a centralised setting with more advanced `know your customer' techniques) can be attained.

\subsection{Contributions}
We study a problem where reviews are crowdsourced from a set of agents. These agents, based on their skills, chose the proposals for which they will write a review, as well as the quality that their reviews will have. In this setting, a mechanism takes the decision of the agents as an input, and rewards them according to the quality of their efforts. We stress that the agents actions are not to report their skills or any private information (as this would be infeasible in a distributed setting, as explained above), but the \emph{reviews themselves}. Given this restriction, we are interested in the design of mechanisms that always have pure Nash equilibria, while in addition, we desire these equilibria to have good performance with respect to certain meaningful objectives.

Our contribution can be summarized as follows: We focus on a very simple mechanism for this problem, study its equilibria under two different scenarios, and investigate their performance with respect to the cumulative review quality (\emph{quality} objective), and the number of proposals that acquire at least one review (\emph{coverage} objective). More specifically:

\begin{itemize}
    \item In the first scenario (Section \ref{sec:1prop}), we assume that there is only one proposal, and the agents can choose between writing an excellent, a good, or no review at all. We prove, that our mechanism always has pure Nash equilibria in this case, and we show that at every such equilibrium, an $\frac{1}{4}$-approximation to the optimal \emph{quality} objective is guaranteed.
    \item In the second scenario (Section \ref{sec:mprop}), we assume that there are multiple proposals and every agent can write a review for any subset of proposals they want. This time however, their choice is limited between writing a review or not. %However, their overall `work' needs to respect a \emph{time constraint}, otherwise this case would reduce to the first one.
    For this scenario, we show that pure Nash equilibria always exist, and all such equilibria provide a $\frac{1}{3}$-approximation to the optimal \emph{coverage} objective. However, this guarantee only holds if our mechanisms has access to \emph{twice} the budget for rewards that the optimal has. Otherwise, we show that there are instances where the coverage achieved at any equilibrium is arbitrarily bad.
\end{itemize}
In both cases, the guarantees are provided against an optimal algorithm that knows everything about the agents (including how capable they are at reviewing each proposal) and rewards each of them with just enough to cover their cost. Our mechanisms never learns anything private: the agents just write the reviews directly and share the rewards following our rules.

We want to point out that, as mentioned in the introduction, the formulation of our model is based on real life scenarios and actual challenges that blockchain initiatives have to face. The same goes for the motivation behind the mechanism we are focusing on, as a version of it is currently used by Cardano's treasury in order to resolve such problems. Our main goal is to investigate which are the theoretical guarantees that Project Catalyst provides under the aforementioned scenarios, and whether they match its real life performance (some examples of which are presented in Section \ref{sec:RD}). To our knowledge, the model that we consider has not been explored so far, although similar problems have been studied in the context of profit sharing games, crowdsourcing through contests, and budget feasible mechanism design. We provide a summary of indicative related papers from these areas, along with a discussion regarding similarities and differences in the next section.

%%%%%%%%%%%%%%%%%%%%%%%%%%%%%
\subsection{Related Work}
%%%%%%%%%%%%%%%%%%%%%%%%%%%%%

Our work is closely related to the area of crowd-sourcing through contests \cite{V2016}. In this type of problems, there is usually  an organizer that announces some tasks that need to be completed, along with a reward for the winning contestant(s). Then the contestants submit their solutions and the organizer decides who will be the winner(s). The goal is usually the design of mechanisms, the stable states of which provide good performance guarantees according to the desired objective. It is an area of problems that is also connected to the one of all-pay auctions, as indicated by several papers \cite{ChawlaHS19,luo2016incentive,siegel2009all,DiPalantinoV09}, while there are also some parallels with mining for proof-of-work cryptocurrencies such as Bitcoin \cite{arnosti2018bitcoin}. Although there is vast amount of works in the area, ours has several differences with most of the related literature. In particular, Chan et al.~\cite{ChanPL20} study the Price of Anarchy in Tullock contests in a setting that is similar to ours, with some of the main differences to be that each agent can choose one contest to participate in, and the behavior of the agents is not affected by time constraints as in our case.  The problem that they consider is also similar to the one presented in \cite{luo2016incentive}, although the latter regards the incomplete information setting. For other works that consider the incomplete information setting on problems of contest design, some relevant papers are from Chawla et al.~\cite{ChawlaHS19}, Moldovanu and Sela~\cite{MS01}, Elkind et al.~\cite{elkind2021contest}, and Glazer and Hassin~\cite{GH88}. Besides the difference in the information setting, these works along with other examples \cite{T95}, focus on objectives that are different than ours. One common difference is that typically, the reward paid by the designer is subtracted from the objective, and the quality of work produced by the participants can be measured accurately. In addition, our model also differs on how the strategy space of the agents is defined, as we assume that it is discretized, given the coarse, crowd-sourced nature of the evaluations. For papers that study objectives that are more aligned to ours, the reader should consider the works of Archak and Sundararajan~\cite{ArchakS09}, Gavious and Minchuk~\cite{GaviousM14} and Edith et al.~\cite{elkind2022simultaneous}, who also study the setting from a complexity standpoint.

  The problem of crowd-sourcing agents (subject to budget constraints) so that certain tasks are completed, has also been studied in the context of truthful budget feasible mechanism design, initiated by Singer~\cite{Singer10}. In these problems there is usually a single buyer, that wants to hire a set of workers that are able to complete certain tasks. Each worker can perform a single task and has a cost for it, which is her private information, while the buyer has a valuation function defined over the power set of the set of tasks. Although similar in principle to the problem that we study, there are also some key differences on how it is approached. In particular, the workers in this model declare their costs for performing the tasks, and the produced outcome is based on these declarations (they are the input to the designed mechanism). In contrast, in our case the agents can freely choose which tasks (reviews) they will complete, and each worker is not bounded by completing just a single task. In addition, the goal in budget feasible mechanism design, is the design of truthful (from the side of the workers) mechanisms that respect the budget constraints, and maximize the valuation function of the buyer. In  our case, the objectives for which we require good performance are of different nature, and they can also be affected by the effort that each agent decides to spend on the tasks that she performs (i.e., the strategy space is different). For more on this topic, we refer the reader to the following indicative works \cite{LeonardiMSZ21,AnariGN14,GoelNS14,AmanatidisBM16,AmanatidisBM17,GravinJLZ19}.
  
  Finally, we want to mention that similar problems have been considered in the context of profit sharing games \cite{AnshelevichP16,AugustineCEFGS15,BachrachSV13}, utility games \cite{Vetta02,GollapudiKPP17}, and project games \cite{BiloGM19,PolevoyTW14}. In most of the aforementioned examples, agents select a set of projects-tasks to participate in, and then they are rewarded according to their contribution. The goal is once again the design of mechanisms that perform well in their stable states. Although there are several differences to the setting that we study, i.e., strategy space of the agents, targeted objective for maximization, the way that the utility of an agent is defined (usually it is assumed that an agent's utility depends only to the reward that she gets, while the cost that she has for performing the task is not considered), etc, most of the related literature studies (among others) the performance of simple mechanisms that proportionally allocate the rewards, and thus follows an approach that is quite similar to ours.

%%%%%%%%%%%%%%%%%%%%%%%%%%%%
\section{Preliminaries}
%%%%%%%%%%%%%%%%%%%%%%%%%%%%

Let $V = \{1, 2, 3, \ldots, n\}$ be the set of community advisors (CA's), $P = \{P_1, P_2, P_3, \ldots, P_m\}$ be the set of proposals, and $\budget \in \mathbb{R}_{> 0}$ the total available budget that is distributed equally among the proposals. 
 The role of the CA's is to write reviews for the proposals. These reviews can be of varying quality, indicated by $q \in \mathcal{Q} = \{0, 1, \ldots, Q\}$, with a quality $q=0$ encoding that a CA did not submit a review for the respective proposal and $Q$ being the highest possible quality.
 
 In total, the strategy of CA $i$ is selecting a vector $\vecc q_i=[q_{i1}, q_{i2}, \ldots,\allowbreak q_{im}]$, where $q_{ij}$ is the quality of the the review that she will write for a proposal $j$. We combine all strategy vectors to form the strategy matrix $\vecc q = (\vecc q_1, \vecc q_2, \ldots, \vecc q_n)$. We also use $\vecc q^j=[q_{1j}, q_{2j}, \ldots, q_{nj}]$ to denote the qualities of the reviews that proposal $j$ gathered from the CA's at $\vecc q$. Moreover, we use function $f:[Q]\rightarrow \mathbb{R}_{\ge 0}$, to describe how much work is needed to write a review of a specific quality. We specify here that $f(\cdot)$ is a strictly increasing function on $q$, and $f(0)=0$ as not submitting a review takes zero effort. Note that function $f$ is CA independent. 
 In addition, every CA $i$ has a private vector of positive parameters, $\vecc s_i=[s_{i1}, s_{i2}, \ldots, s_{im}]$, that describe the \emph{time cost} (e.g., in seconds) required to produce a review of quality $q=1$ for each proposal. Often, we will refer to the agents' skills: the higher these skills are, the lower the $s_{ij}$ necessary to write a review. Therefore, $f(q_{ij})\cdot s_{ij}$ can be seen as the cost that CA $i$ has for writing a review of quality $q_{ij}$ for proposal $j$\footnote{For a crude analogy, imagine that $f(q)$ is the number of words required for a review of quality $q$ and $s_{ij}$ is the minutes-per-word possible by CA $i$ for proposal $j$.}. Note that the skills $s_{ij}$ are only known to the CA's. As with $\vecc q$, we use $\vecc s$ for the matrix of all skill vectors. Since the reviewing period has a fixed duration $\timeLimit$, every CA has to spend their effort within this timeframe. We call this the \emph{maximum effort} constraint. Formally, for every $i \in V$:
\begin{equation}\label{def:max-effort}
	\sum_{j \in P} f(q_{ij}) \cdot s_{ij} \le \timeLimit.
\end{equation}

%\textcolor{red}{One alternative that perhaps seems more natural, is to limit the \emph{number} of reviews submitted by each CA. However, this is unsuitable. It cannot be enforced in a blockchain system, because CA's are pseudonymous and a suitably motivated person could circumvent this by creating multiple identites. In addition, because of different skill levels and dedication between CA's it is not a natural constraint that every person behind the CA identities would only write up to a specific number of reviews.}

A mechanism $\mech$ takes as input \emph{only} the CA's strategies $\vecc q$ (\emph{without} the individual skills) and outputs a payment vector $\vecc p = \vecc p(\vecc q) \in \mathbb{R}^n_{\ge 0}$. We require that the mechanism is budget feasible. More specifically, as the budget is distributed equally among the proposals, for every $j \in P$, we desire:
\begin{equation}\label{def:budget-feasible}
    \sum_{i=1}^n p_{ij}(\vecc q) \le \frac{\budget}{m}=\beta
\end{equation}
Given the strategy vector $\vecc q$, the utility of CA $i$ is:
\begin{equation}
    u_i^{\mech}(\vecc q) = p_i(\vecc q) 
    - \sum_{j=1}^m s_{ij} \cdot f(q_{ij}),
\end{equation}
that is, the payment received given the actual quality of the submitted reviews minus the effort of producing them. Given a mechanism $\mech$, we say that $\vecc q$ is a pure Nash equilibrium of $\mech$ if, for every CA $i$ and possible deviation $\vecc q_i' \in \mathcal{Q}$, we have that:
\begin{equation}
    u_i^{\mech}(\vecc q) \ge u_i^{\mech}(\vecc q_i', \vecc q_{-i}).
\end{equation}
We use $\vecc q \in \texttt{Feasible}(\vecc s, \budget, \timeLimit)$ if \Cref{def:max-effort} and \Cref{def:budget-feasible} are satisfied. %(i.e., if the outcome is within budget and no CA exceeds their working capacity).
Moreover, we use $\vecc q \in \texttt{PNE}(\vecc s, \budget, T)$ for the set of outcomes that are the pure Nash equilibria of mechanism $\mech$, in addition to being feasible. %\textcolor{red}{Say something about the corner cases of being indifferent between writing and not writing a review.} %It is easy to see that $\texttt{PNE}(\vecc s, \budget, T) \subseteq \texttt{Feasible}(\vecc s, \budget, \timeLimit)$ 

%%%%%%%%%%%%%%%%%%%%%%%%%%%%
\subsection{Objectives}
%%%%%%%%%%%%%%%%%%%%%%%%%%%

The ideal mechanism $\mech$ would, at a pure Nash equilibrium, ensure that the final outcome provides a good cumulative review quality. In addition it  should also motivate the CA's to write reviews in a way, so that the number of proposals that receive at least one review is as big as possible. We formally define these two objectives below.

\paragraph{Quality.}
It is desirable that the total cumulative review quality is maximized:
\begin{equation}
    \texttt{Qual}(\vecc q) = \sum_{i \in V} \sum_{j \in P} q_{ij}.
\end{equation}

\paragraph{Coverage.}
Additionally, it is desirable to maximize the number of proposals that acquire at least one review:
%whose reviews surpass some cumulative quality $Q_{\min} \ge 1$, which is defined as the highest cumulative quality for which it is possible to have at least $1$ review per proposal. Formally $Q_{\min}$ is the largest value such that there exists $\vecc q$ with the following property:
%\begin{equation}
    %\text{For every proposal } j \in P \text{ we have that} %\sum_{i \in V} q_{ij} \ge Q_{\min},
%\end{equation}
%while respecting the CA's maximum effort \eqref{def:max-effort} and remaining budget feasible \eqref{def:budget-feasible}. If such a value does not exist, we set $Q_{\min} = 1$. Given this $Q_{\min}$, the \emph{coverage} objective is:
\begin{equation}
\texttt{Cov}(\vecc q) = \left|\left\{j \in P \;|\; \sum_{i \in V} q_{ij} \ge 1\right\}\right|\footnote{Notice that the minimum quality of a review is 1.}.
\end{equation}

%%%%%%%%%%%%%%%%%%%%%%%%%%%%%%%%%%%%
\subsection{Price of Anarchy}
%%%%%%%%%%%%%%%%%%%%%%%%%%%%%%%%%%%%

To quantify the performance of a mechanism for an objective (in our case either quality or coverage), we want to compare the outcome that is produced at its worst pure Nash equilibrium, with the optimal possible outcome. For the optimal outcome we consider the non-strategic version of the problem, i.e., which would be the optimal solution if we were able to purchase reviews directly, at exactly the cost required to produce them by the corresponding agents \footnote{In the strategic setting on the other hand, there is some inefficiency because the agents participating in this procedure have incentives, and thus want to maximize their rewards.}. In that sense, when someone considers the non-strategic version of the problem, the budget feasibility constraint of \Cref{def:budget-feasible}, can be translated as follows: We say that a vector $\vecc q$ respects the budget feasibility requirements if for every $j \in P$, we have

\begin{equation}\label{def:offbudget-feasible}
    \sum_{i=1}^n s_{ij} \cdot f(q_{ij}) \le \frac{\budget}{m}=\beta.
\end{equation}
We use $\vecc q \in \texttt{nsFeasible}(\vecc s, \budget, \timeLimit)$ if \Cref{def:max-effort} and \Cref{def:offbudget-feasible} are satisfied. %\textcolor{red}{Added this part, the reviewer is kind of right as we have not defined the budget feasibility constraint when we consider the offline solution (when there are no payments). Check again if everything is clear now.} 
With these in hand, we can use the notion of Price of Anarchy (PoA) \cite{KP09} in order to quantify the performance of a mechanism. In particular, given $V, P$, and $f$, the PoA of mechanism $\mech$ for an objective \texttt{Obj} is defined as:
\begin{equation}
\texttt{PoA}(\mech) = 
\max_{\vecc s}
\max_{\vecc q}
\min_{\vecc q'}
\frac{\texttt{Obj}(\vecc q)}{\texttt{Obj}(\vecc q')},
\end{equation}
where $\vecc q \in \texttt{nsFeasible}(\vecc s, \budget, \timeLimit)$, and $\vecc q' \in \texttt{PNE}(\vecc s, \budget, T)$. Intuitively, given an objective, this metric describes the ratio between the optimal outcome, over the worst case equilibrium, for the worst possible selection of skills $\vecc s$. At a high level: we select a mechanism $\mech$ and then an adversary creates a group of reviewers with skill $\vecc s$ in order to embarrass us, by comparing the objective value of the worst equilibrium of our mechanism (i.e. $\vecc q'$) to the best possible outcome $\vecc q$, for these $\vecc s$. 
 
\paragraph{Resource Augmentation.} Sometimes, as PoA results can be overly pessimistic, it might be worth investigating if they can be improved, when a mechanism has access to additional resources (in our case, to a larger budget). 

Specifically, we would like to explore how the optimal outcome, for a specific objective, compares with the worst possible outcome that a mechanism produces at a pure Nash equilibrium, when it has access to a larger budget. To this end, we introduce the notion of $\texttt{PoA}_{\texttt{k}}$, which is formally defined as follows:
\begin{equation}
\texttt{PoA}_{\texttt{k}}(\mech) = 
\max_{\vecc s}
\max_{\vecc q}
\min_{\vecc q^{\star}}
\frac{\texttt{Obj}(\vecc q)}{\texttt{Obj}(\vecc q^{\star})},
\end{equation}
where $\vecc q \in \texttt{nsFeasible}(\vecc s, \budget, \timeLimit)$, and $\vecc q^{\star} \in \texttt{PNE}(\vecc s, k \cdot \budget, T)$, for $k \geq 1$. %We want to point out that if $\vecc q \in \texttt{Feasible}(\vecc s, \budget, \timeLimit)$, then $\vecc q \in \texttt{Feasible}(\vecc s, k\cdot \budget, \timeLimit)$. 

%%%%%%%%%%%%%%%%%%%%%%%%%%%%%%%
\subsection{A Simple Proportional Mechanism}
%%%%%%%%%%%%%%%%%%%%%%%%%%%%%%%
As we have already mentioned, a mechanism in this setting simply takes as input the strategy vector $\vecc q$ of the CA's, and outputs a payment vector $p(\vecc q)$. In this work, we focus on the study and analysis of a robust and intuitive mechanism that essentially rewards the CA's in a proportional manner. In particular, recall that each proposal $j$ has an available budget $\beta$, that eventually will be distributed to CA's that write a review for it. In addition, as we have pointed out, the function $f(\cdot)$ defines how much work a review of a specific quality needs. We are interested in mechanisms that reward the agents in a \emph{fair} way, depending on the quality of the reviews produced (and by implication, the work they put in). To this end, we study the following proportional reward scheme, that also defines the mechanism that we focus on: 

\begin{equation}
    p_i(\vecc q) = \sum_{j \in P} 
    \beta \cdot
    \frac{f(q_{ij})}{\sum_{i \in V} f(q_{ij})}
\end{equation}
Intuitively, this is a mechanism with which we try to `match' the difficulty of producing a review of a specific quality, by providing a proportional reward for each case.

For the remainder of the paper, we refer to this mechanism as the \emph{proportional mechanism}, and we focus on the existence and the performance of its equilibria (with respect to the objectives that we defined), under two different scenarios. From this point on-wards, we will also refer to the CA's as agents.

\subsection{The Choice of Function $f(\cdot)$}\label{sec: func}
The function $f(\cdot)$ dictates the structure of the equilibria produced by the mechanism, depending on our goal, which could range from only requiring reviews of the highest quality to enabling players of any skill level to contribute. This can be tuned by the designer, who can show the agents reviews of different qualities with their associated `difficulties' $f$. We develop some intuition through a couple examples, illustrating some reasonable choices. The simplest example is a linear function $f(q) = q$, which is the one actually used for Project Catalyst, as it provides a good balance between the higher difficulty of writing good reviews and the rewards provided, giving a sense of `fairness' to the participants. Let us compare this with another possible choice, the exponential $f(q) = 2^q$. This choice allows fewer agents to write high quality reviews, as only the most skilled ones would have be able to write such a review without violating their time constraint. On the flip side, these highly skilled agents would be highly motivated to provide those reviews, collecting the majority of rewards, while lower skilled agents would have to wait and find `gaps' in proposals, in order to contribute lower quality reviews and collect the remaining rewards. On the contrary, a concave choice such as $f(q) = \log (q)$ skews this balance towards lower skilled agents, which now have the power to write higher quality reviews with less effort.

\section{A Single Proposal: The Quality Objective}\label{sec:1prop}
%%%%%%%%%%%%%%%%%%%%%%%%%%%%%%%%%%%%%%%%%%%%%%%%

In this section, we study the proportional mechanism under the simple scenario where there is only one proposal, and for every $i \in V$, we have that $q_{i} \in \{0, 1, \alpha\}$, for some $\alpha \geq 2$. This set of possible qualities can be interpreted as writing no review, writing a 'good' review, and writing an 'excellent' review respectively. Notice that the intuition behind our formulation is that for the designer, an excellent review is $\alpha$ times better than a good one. Since there is only one proposal, the interesting question under this scenario is to explore how our mechanism performs under the quality objective. So the first step is to examine whether this mechanism always has pure Nash equilibria.

Before we begin, we point out that for the remainder of this section we rename the agents so that $s_{i} \le s_{i+1}$ for every $i$. Finally, we choose function $f(\cdot)$ to be:

\[	f(q_i)=  
	\begin{cases}
	0, & \text{if $q_i=0$} \\
	1, & \text{if $q_i=1$}  \\
	\alpha, & \text{if $q_i=\alpha$}
	\end{cases} 
\]
for every $i \in V$, and $\alpha \geq 2$, indicating that the effort for writing a review is analogous to its quality\footnote{We note here for completeness that in the special case where every agent decides to not write a review, i.e., $q_i=0$ for every $i$, the proportional reward scheme is not applied and simply none of the agents is rewarded, i.e., the reward of everyone is 0. }. As mentioned in Section \ref{sec: func}, this function is used by Project Catalyst, and captures how much work is needed to write a review of a specific quality, given the provided guidelines. %\textcolor{red}{ Redefined the function, added this small explanation.}

%%%%%%%%%%%%%%%%%%%%%%%%%%%%%%%%%%%%%%%%%%%%%%%%%%%%%%%
\subsection{A Single Proposal: Existence of Pure Nash Equilibria}\label{sec:prop1}
%%%%%%%%%%%%%%%%%%%%%%%%%%%%%%%%%%%%%%%%%%%%%%%%%%%%%%%

For this particular case, of one proposal and three possible strategies, we introduce a small modification that the proportional mechanism needs, so that the existence of pure Nash equilibria can always be guaranteed\footnote{This small calibration was not demonstrated in the preliminary version of our paper \cite{BirmpasKLO22}.}. 

\begin{equation}
    p_i(\vecc q) = 
    \beta \cdot
    \frac{f(q_{i})}{\sum_{j \in V} f(q_{j})+\sqrt{\alpha}}
\end{equation}

We call this modified version of the mechanism, the \emph{modified} proportional mechanism. Notice that the only difference is the addition of $\sqrt{\alpha}$ in the denominator, which implies smaller rewards for the agents, and that part of the budget will not be allocated. This modification creates a (needed) better balance between the rewards that an agent will get by following different strategies, so that the existence of pure Nash equilibria, as we will see in the theorem that follows, is not affected.

\begin{theorem}\label{thm:pne}
Consider an instance with only one proposal, and $q_i \in \{0, 1, \alpha\}$ for every $i \in V$. Then, the modified proportional mechanism has always at least one pure Nash equilibrium.
%Moreover, at any pure Nash equilibrium, the number of tickets submitted is at least $i^\star \le n$, where $i^\star$ is the largest index such that:
%\begin{equation}\label{ineq:star}
    %s_{i^\star} \le \frac{1}{i^\star}
%\end{equation}
\end{theorem}
\begin{proof}
 Let $i^\lambda$ to be the largest $i$ for which we have $s_{i} \le \frac{\beta}{i+\sqrt{\alpha}}$. We begin by defining an initial state where every agent in $G_0 = \{i \;|\; i \le i^\lambda\}$ writes a good review, while every agent in $Z_0 = \{i \;|\; i > i^\lambda\}$ does not write a review \footnote{If $G_0$ is empty, then it is easy to see that having no agent writing a review, is a pure Nash equilibrium where the utility of every agent is 0.}. Now notice that the utility of an agent $i \in G_0$ is
\begin{equation}
    \frac{\beta \cdot f(1)}{\sum_{i \leq i^\lambda} f(1)+\sqrt{\alpha}} - s_i \cdot 1 = \frac{\beta}{i^\lambda+\sqrt{\alpha}} - s_i \ge 0,
\end{equation}
as $i^\lambda$ is the largest index for which we have $s_{i} \le \frac{1}{i+\sqrt{\alpha}}$, and $s_i \leq s_{i^\lambda}$ for every $i \in G_0$. In addition, the utility of the agents in $Z_0$ is zero by definition. We will construct a pure Nash equilibrium by repeatedly updating the agents; best responses \footnote{A best response of an agent $i$ is the best strategy that she can follow, given the strategies of the other agents.} given the current set of reviews, starting from $G_0$.

Let $E_t \subseteq V$, $G_t \subseteq V$, and $Z_t \subseteq V$ be the sets of the agents that wrote an excellent, a good, and no review respectively at step $t$. Assume for now that the utilities of the agents of every set is non-negative, and let $g(E_t, G_t, Z_t)$ be a function which refines the current strategies of the agents. Specifically $g(E_t, G_t, Z_t) = (E_{t+1}, G_{t+1}, Z_{t+1})$ such that:
\begin{itemize}
    \item  Let $i'$ be the agent with the minimum index in set $G_t$. If for $i'$ is beneficial to convert her good review to an excellent, then $E_{t+1} = E_t \cup \{i'\}$. Notice that if this conversion is not beneficial for agent $i'$, then it is not beneficial for any $i \in G_t$, as $s_i' \leq s_i$. Moreover, observe that if agent $i'$ is not able to write an excellent review (although it is beneficial for her) due to constraint $T$, then the same holds for any agent in $G_t$.
    \item Let $i_1 > i_2 > \ldots > i_{\alpha-1}$ be the largest $\alpha - 1$ indices in $G_t$. For them, it may be the case that a good review yields a negative utility, now that an excellent review has been added. The reason for this is that after the conversion of one review from excellent to good, the denominator of the reward has been increased by an additional factor of $\alpha-1$. Removing at most $\alpha-1$ agents from $G_t$, choosing the ones with the highest $s_i$'s, is enough to bring the denominator back to what it was before, thus guaranteeing that the newly formed set has only agents with non-negative utility. So, starting from $i_{\alpha-1}$, we remove up to $\alpha - 1$ of these agents from $G_t$ to $Z_t$, one by one until there are no agents with negative utility. Therefore, we obtain $G_{t+1}$, where every agent obtains a non-negative utility from their review, as well as $Z_{t+1}$ which contains every agent in $Z_t$, along with the newly added agents that were removed from $G_t$. 
\end{itemize}
By the definition of function $g(\cdot)$, it is clear that repeated application on the inputs $(E_t, G_t, Z_t)$ would eventually reach a steady state (i.e., a \emph{fixpoint}) $(E_F, G_F, Z_F)$, such that $g(E_F, G_F, Z_F) = (E_F, G_F, Z_F)$. This is clear as after each application, the sets $E_{t+1}$, $Z_{t+1}$ are supersets of sets $E_t$ and $Z_t$ respectively, while $G_{t+1}$ is a subset of $G_t$. Since the number of agents is finite this procedure will eventually stop. 
We will show that the fixpoint of $g(\cdot)$ starting from $(\emptyset, G_0, V\setminus G_0)$ (our initial configuration) is in fact a pure Nash equilibrium. So let $(E_t, G_t, Z_t) = g^t(\emptyset, G_0, V\setminus G_0)$ be the sets resulting from $t$ composed applications of $g$, where $g^0$ is the identity function. 

Using induction, we will first show that for every $t$, the agents in $E_t$ are always playing a best response %(i.e., given the current set of reviews they are maximizing their payoff by submitting an excellent review)
and in addition, $i_e < i_h$ for any $i_e \in E_t, i_h \in G_t$. Clearly, for $t=0$ this is true as $E_0 = \emptyset$ i.e., it does not contain any agent. For the induction hypothesis, we assume that the our claim holds for $t$. Now let $i' = \min_i G_t$ be the agent who deviates to writing an excellent review instead of a good one: by definition, this is her best response. We need to show that for the rest of the agents in $E_{t+1}$, i.e, $E_{t+1} \setminus \{i'\} = E_t$, writing an excellent review remained a best response. To this end, let $i_e  \in E_t$: we will prove that for her, writing an excellent review is better than writing either a good review or no review at all. This will be done in two steps: first we will show that this holds after applying the first part of $g(\cdot)$, before the set $G_t$ is updated. At this point, we compare the current utility of $i_e$ with the one of $i'$:
\begin{equation*}
    \frac{\alpha \cdot \beta}{\alpha \cdot |E_{t+1}| + |G_{t}|+\sqrt{\alpha}-1} - \alpha\cdot s_{i_e}
    \ge 
    \frac{\alpha \cdot \beta}{\alpha \cdot |E_{t+1}| + |G_{t}|+\sqrt{\alpha}-1} - \alpha\cdot s_{i'}
    \ge 0,
\end{equation*}
because $i_e \le i'$ from the inductive hypothesis, while in addition, $i_e \le i' \Rightarrow s_{i_e} \le s_{i'}$ by the definition of the $s_i$. Therefore, writing an excellent review is at least as good as writing nothing for agent $i_e$, thus we just need to compare with the payoff of a good review. Since writing an excellent review is a best response for $i'$:

\begin{align*}
    &\frac{\alpha \cdot \beta}{\alpha \cdot |E_{t}| + |G_{t}|+\sqrt{\alpha}+\alpha-1} - \alpha\cdot s_{i'}
    \ge
    \frac{\beta}{\alpha \cdot |E_{t}| + |G_{t}|+\sqrt{\alpha}} - s_{i'}\\
    \Rightarrow\quad
    &\frac{\alpha \cdot \beta}{\alpha \cdot |E_{t+1}| + |G_{t}|+\sqrt{\alpha}-1} - \alpha\cdot s_{i'}
    \ge
    \frac{\beta}{\alpha \cdot |E_{t+1}| + |G_{t}|+\sqrt{\alpha} - \alpha } - s_{i'}\\
    \Rightarrow\quad
    &\frac{\alpha \cdot \beta}{\alpha \cdot |E_{t+1}| + |G_{t}|+\sqrt{\alpha}-1} - \frac{\beta}{\alpha \cdot |E_{t+1}| + |G_{t}|+\sqrt{\alpha} - \alpha}
    \ge
    (\alpha - 1) \cdot s_{i'} \ge (\alpha - 1) \cdot s_{i_e}\\
    \Rightarrow\quad
    &\frac{\alpha \cdot \beta}{\alpha \cdot |E_{t+1}| + |G_{t}|+\sqrt{\alpha}-1} - \alpha \cdot s_{i_e}
    \ge \frac{\beta}{\alpha \cdot |E_{t+1}| + |G_{t}|+\sqrt{\alpha} - \alpha} - s_{i_e},
\end{align*}
with the last inequality showing that writing an excellent review is a best response for $i_e$.

All that is left to show, is that removing some good reviews (up to $\alpha-1$), will not cause any of the agents in $E_{t+1}$ to deviate to writing a good review instead. For any $i_e \in E_{t+1}$, we have that:
\begin{align}\label{ineq:current_best}
    &\frac{\alpha \cdot \beta}{\alpha \cdot |E_{t+1}| + |G_{t}|+\sqrt{\alpha}-1} - \alpha \cdot s_{i_e} \ge \frac{\beta}{\alpha \cdot |E_{t+1}| + |G_{t}|+\sqrt{\alpha} - \alpha} - s_{i_e} \notag \\
    \iff\quad
    &\frac{\alpha \cdot \beta}{\alpha \cdot |E_{t+1}| + |G_{t}|+\sqrt{\alpha}-1}- \frac{\beta}{\alpha \cdot |E_{t+1}| + |G_{t}|+\sqrt{\alpha} - \alpha} \geq \alpha \cdot s_{i_e}-s_{i_e}
\end{align}
Our goal now will be to show that: 
\begin{align}\label{ineq:current_best1.5}
  &\frac{\alpha}{\alpha \cdot |E_{t+1}| + |G_{t+1}|+\sqrt{\alpha}}- \frac{1}{\alpha \cdot |E_{t+1}| + |G_{t+1}|+\sqrt{\alpha} - \alpha +1}\notag\\ \geq &\frac{\alpha}{\alpha \cdot |E_{t+1}| + |G_{t}|+\sqrt{\alpha}-1}- \frac{1}{\alpha \cdot |E_{t+1}| + |G_{t}|+\sqrt{\alpha} - \alpha}
  \end{align}
  which means
  \begin{align}\label{ineq:current_best2}
    &\frac{\alpha \cdot \beta}{\alpha \cdot |E_{t+1}| + |G_{t+1}|+\sqrt{\alpha}}- \frac{\beta}{\alpha \cdot |E_{t+1}| + |G_{t+1}|+\sqrt{\alpha} - \alpha +1}\notag \\ \geq &\frac{\alpha \cdot \beta}{\alpha \cdot |E_{t+1}| + |G_{t}|+\sqrt{\alpha}-1}- \frac{\beta}{\alpha \cdot |E_{t+1}| + |G_{t}|+\sqrt{\alpha} - \alpha},
\end{align}
as by combining \eqref{ineq:current_best} and \eqref{ineq:current_best2} will give us,
\begin{align} \label{ineq:induction}
    &\frac{\alpha \cdot \beta}{\alpha \cdot |E_{t+1}| + |G_{t+1}|+\sqrt{\alpha}} - \alpha \cdot s_{i_e} \ge \frac{\beta}{\alpha \cdot |E_{t+1}| + |G_{t+1}|+\sqrt{\alpha} - \alpha+1} - s_{i_e}
\end{align} 
which shows that deviating to a writing a good review is not beneficial. To do so, we will study the monotonicity of function 
\begin{equation*}
   f(x)=\frac{\alpha}{x+\sqrt{\alpha}} - \frac{1}{x+\sqrt{\alpha} - \alpha + 1}
\end{equation*}
for $\alpha \geq 2$, and $x \ge \alpha$ (as we currently assume that there is at least one agent that writes an excellent review). Setting $y=x+\sqrt{\alpha}$, we have the following:
\begin{align*}
    &-\frac{\alpha}{y^2}+\frac{1}{(y-\alpha+1)^2}=0\\
    \Rightarrow\quad
    & y^2=\alpha \cdot(y-\alpha+1)^2\\
     \Rightarrow\quad
    & (\alpha-1)\cdot y^2 -2\cdot\alpha\cdot(\alpha-1)\cdot y +\alpha\cdot (-\alpha+1)^2=0 \\
     \Rightarrow\quad
    &  y^2 -2\cdot\alpha\cdot y -\alpha\cdot (\alpha-1)=0
\end{align*}
For the solutions of this equation, we have that $y_{1}=\alpha+\sqrt{\alpha}$ and $y_{2}=\alpha-\sqrt{\alpha}$, which means $x_1=\alpha$ and $x_2=\alpha-2\cdot\sqrt{\alpha}$. So we are only interested in $x_1$ (as $x\geq \alpha$). It is easy to see that function $f(x)$, for $x\geq x_1$ is decreasing. The latter guaranties that if $\alpha\leq z_1 \leq z_2$, then

\begin{equation*}
\frac{\alpha}{z_1+\sqrt{\alpha}} - \frac{1}{z_1+\sqrt{\alpha} - \alpha + 1} \geq \frac{\alpha}{z_2+\sqrt{\alpha}} - \frac{1}{z_2+\sqrt{\alpha} - \alpha + 1}
\end{equation*}

Now recall that $|G_{t+1}|\leq|G_{t}|-1$ as we know that an agent moved from writing a good review to writing an excellent review (and probably at most $\alpha-1$ agents left after that). Therefore, we have $\alpha \cdot |E_{t+1}|+|G_{t+1}|\leq \alpha \cdot |E_{t+1}|+|G_{t}|-1$, and $|E_{t+1}|\geq 1$. By setting $z_1=\alpha \cdot |E_{t+1}|+|G_{t+1}|\leq \alpha \cdot |E_{t+1}|+|G_{t}|-1=z_2$ we have that relation \eqref{ineq:current_best1.5} always holds, and therefore, we get \eqref{ineq:induction}.

Now let us argue about the agents in $G_t$. For this set of agents is sufficient to show \footnote{The reason for this is that when we reach set $G_F$, we know that either it is not beneficial for any of the agents of this set to write an excellent review, or they cannot due to constraint $T$.} that for every $t$ their utility is non-negative, thus they do not want to move to set $Z_t$. Again, by using induction, the statement trivially holds for set $G_0$ by definition. By the induction hypothesis we know that this is also true for $G_t$. After the application of function $g(\cdot)$, the minimum indexed agent will write an excellent review instead of a good one, and this may make at most $a-1$ agents from $G_t$ to end up with a negative utility. However, recall that set $G_{t+1}$ is formed by moving this set of agents to $Z_t$ (thus forming $Z_{t+1}$). Therefore, set $G_{t+1}$ is a subset of $G_t$ that only contains agents with non-negative utility.

Finally, the last set that remains is $Z_t$. We want to show that for any $t$, the agents in $Z_t$ will have a non-positive utility if they try to write either an excellent, or a good review. Initially, notice that in general, an agent $i$ that writes no reviews and at the same time has a negative utility by writing a good review, she has also an negative utility by writing an excellent review. The reason is that if $$\frac{\beta }{\alpha \cdot |E|+|G|+\sqrt{\alpha}+1} - s_i<0,$$ this implies that 
\begin{equation*}
  \frac{\beta}{\alpha \cdot |E|+|G|+\sqrt{\alpha}+\alpha} - s_i<0
  \Rightarrow
  \frac{\alpha \cdot \beta }{(\alpha+1) \cdot |E|+|G|+\sqrt{\alpha}} - \alpha \cdot s_i<0,  
\end{equation*}
%$\frac{\alpha \cdot \beta}{\alpha|E|+|G|+1} - \alpha \cdot s_i<0 \Rightarrow \frac{\alpha \cdot \beta }{(\alpha+1)|E|+|G|} - \alpha \cdot s_i<0$, 
where $E, G$, the sets of agents that were currently writing excellent and good reviews respectively. We proceed in proving the statement inductively. For $t=0$, no agent $i \in Z_0$ wants to write a good review, as by definition $$\frac{\beta }{i+\sqrt{\alpha}} - s_i<0$$ for every $i >i^\lambda$. Thus, they do not want to deviate to writing an excellent review as well. Now observe that the reward to an agent decreases as $t$ grows (as the denominator of the reward increases each time an agent writes an excellent instead of a good review). Set $Z_{t+1}$, consists of the agents from $Z_t$, and the agents that were removed from $G_t$ since they did not want write a good review. By using the inductive hypothesis and the observation that the reward at step $t+1$ is smaller than the reward at step $t$, it is easy to see that the agents of $Z_{t+1}$ by deviating to either writing an excellent or a good review, derive a negative utility.

Putting everything together, we know that at the fixpoint $(E_F, G_F \allowbreak, Z_F)$ all the agents are playing their best response. Therefore, the proportional mechanism always has at least one pure Nash equilibrium. %Moreover, at every step the total number of tickets claimed can only rise, as one excellent review is added for at most $\alpha$ good ones that are removed.
\end{proof}

\begin{remark}
\normalfont We would like to point out that the pure Nash equilibrium described by the procedure presented in the proof of Theorem \ref{thm:pne} is not unique. In particular, consider the following example: Suppose that we have an instance with 2 agents and skills $s_1=\frac{1}{4}-\epsilon$ for some $0<\epsilon<<\frac{1}{100}$, and $s_2=\frac{1}{9}$. In addition, let $\alpha=4$, $B=1$, and $T=\frac{1}{2}$. For this instance, there are at least two pure Nash equilibria, and the quality guaranties that they provide are different. Initially, notice that due to constraint $T$, agent 1  cannot write an excellent reviews as $4\cdot(\frac{1}{4}-\epsilon)>\frac{1}{2}$. Thus, in any pure Nash equilibrium the only available option for her is to either write a good review, or to not write a review at all. We proceed in describing two different pure Nash equilibria. 
\begin{itemize}
    \item For the first one, consider the state where every agent writes a good review. The utility of agent 1 is $\frac{1}{4}-\frac{1}{4}+\epsilon>0$, therefore this is a best response for her. Regarding agent 2, her utility is $\frac{1}{4}-\frac{1}{9}>0.138>0$, so the only meaningful deviation for her is to write an excellent review. In that case their utility becomes $\frac{4}{7}-\frac{4}{9}<0.127$, and therefore this deviation is not profitable. We conclude that this is a pure Nash equilibrium that provides a total quality of 2.
    \item For the second one, consider the state where agent 2 writes an excellent review, while agent 1 does not write a review. The only possible deviation of agent 1 is to write a good review. In this case her utility becomes $\frac{1}{7}-\frac{1}{4}+\epsilon<0$, thus there is no incentive in doing so. The utility of agent 2 at this state is $\frac{4}{6}-\frac{4}{9}=2/9$. Therefore the only meaningful deviation for each one them is to write a good review. In that case, their utility becomes $\frac{1}{3}-\frac{1}{9}=2/9$. So once again, this deviation is not profitable as the utility remains the same. We conclude that this is a pure Nash equilibrium that provides a total quality of 4.
\end{itemize}

\end{remark}

As this version of the mechanism always guarantees existence of pure Nash equilibria, in the next section we explore the quality of these equilibria under the \emph{modified} proportional mechanism.

%%%%%%%%%%%%%%%%%%%%%%%%%%%%%%%%%%%%%%%%%%
\subsection{Quality Approximation Guarantees}
%%%%%%%%%%%%%%%%%%%%%%%%%%%%%%%%%%%%%%%%%%%

 We proceed in studying how the \emph{modified} proportional mechanism performs at its pure Nash equilibria with respect to the quality objective. For this, we will first need to explore the properties of the optimal outcome when the constraint $T$ does not exist. Given our assumption that an excellent review is $\alpha$ times more useful than a good review, we could first define the optimal outcome of the problem, without the $T$ constraint, as follows:
\begin{definition}[Optimal Outcome]
    Given a budget of $\beta$ and agents with skills $s_1 \le s_2 \le \cdots \le s_n$, the optimal outcome, when constraint $T$ does not exist, is:
    $$
    \begin{array}{ll@{}ll}
        \text{maximize}  & \alpha\cdot|E| + |G|&\\
        \text{subject to}& \sum_{i \in E} \alpha \cdot s_i + \sum_{i \in G} s_i &\le \beta,
    \end{array}
    $$
where $E$ and $G$ are the sets of agents that write excellent and good reviews respectively.    
\end{definition}

 We continue with the following lemma, that shows how the optimal outcome can be computed.

\begin{lemma}\label{lem:greedy}
The optimal outcome can be found using the following greedy algorithm:
\begin{itemize}
    \item Starting from $s_1$, add excellent reviews until no more can be added without exceeding the budget.
    \item Fill the remaining budget with good reviews.
    %\item Both the previous steps should respect the time constraint $T$.
\end{itemize}
\end{lemma}
\begin{proof}
 Assume that sets $E^\star$ and $G^\star$ are the sets of agents that write excellent and good reviews at the optimal solution respectively. Let $j'= \min_i G$ and $i'= \max_i E$. Now suppose that $j'< i'$. Then, this would imply that by exchanging $i'$ and $j'$, a solution of the same value would be produced, but with a lower cost, as 

\begin{align*}
\alpha\cdot \sum_{i\in E^\star \setminus \{i'\}} s_i + &\alpha \cdot s_{i'} + \sum_{i \in G^\star \setminus \{j'\}} s_i - s_{j'}\\
&\leq
\alpha\cdot \sum_{i\in E^\star\setminus \{i'\}} s_i + \alpha \cdot s_{j'} + \sum_{i \in G^\star\setminus \{j'\}} s_i - s_{i'}.
\end{align*}
Repeated applications of the same procedure, either will lead to a contradiction (as the cost of the solution decreases each time), or will eventually give us the greedy solution.
\end{proof}

\begin{remark}
\normalfont {Notice that the same algorithm computes the optimal outcome even if there is a constraint $T$ by following the same arguments.}
\end{remark}

We can now use the previous lemma, in order to upper bound the value of the optimal solution, when the constraint $T$ is not present.

\begin{lemma}\label{lemma:greedy}
For any set of skills $\{s_1, \ldots, s_n\}$, the value of the optimal solution is at most $\alpha \cdot (i^\star + 1)$,
where $i^\star$ is the largest index such that $$\sum_{i=1}^{i^\star} \alpha \cdot s_i \le \beta.$$
\end{lemma}
\begin{proof}
Initially notice that the value of optimal solution is always the same, regardless of the way it is computed. Consider the optimal solution produced by the algorithm described above, and let $E^g$ and $G^g$ to be the sets of agents that write excellent and good reviews in this solution respectively. Now suppose that $|G^g| \ge \alpha$ and let $j' = \min_i G^g$. Then, the sets $E^g \cup \{j'\}$ and $G^g\setminus\{j', j'+1, \ldots, j' + \alpha - 1\}$ represent a solution with the same value but lower cost as:
\begin{align*}
\alpha\cdot \sum_{i\in E^g} s_i + \alpha s_{j'} + &\sum_{i \in G^g} s_i - s_{j'} - s_{j'+1} - \ldots s_{j'+\alpha-1}\\
&=\alpha \cdot \sum_{i\in E^g} s_i + \sum_{i \in G^g} s_i - \alpha \cdot s_{j'} - \sum_{i=j'}^{j'+\alpha-1}s_i\\
&\le \alpha\cdot \sum_{i\in E^g} s_i + \sum_{i \in G^g} s_i.
\end{align*}
Notice that this is a contradiction due to how the greedy algorithm works. Therefore, there are at most $\alpha - 1$ good reviews in an optimal solution. Thus, we can conclude that the optimal solution produced by the greedy algorithm is of the form $E^g = \{1,2, \ldots, i^\star\}, G^g = \{i+1,i+2, \ldots, i+j\}$, where $j \le \alpha - 1$.

\end{proof}

\begin{remark}\label{rem:zo}
\normalfont Notice that if $i^\star=0$, which means that for every $i$ we have that $\alpha \cdot s_i> \beta$, this implies that the value of the optimal solution is at most $\alpha$, something that is consistent with the statement of \Cref{lemma:greedy}. To see this, observe that if the optimal solution's value was greater than $\alpha$, since it is not possible for any agent to write an excellent review (due to the budget constraint), this value would be attained by agents that write reviews of good quality. The latter would mean that $|G^\star|>\alpha$, and $\sum_{i \in G^\star} s_i \leq \beta$. This implies that $\alpha \cdot \min_{i \in G^\star} s_i \leq \beta$, a contradiction.
\end{remark}

We proceed with the following theorem where we bound the PoA of the modified proportional mechanism, with respect to the quality objective. At several steps throughout the proof, we compare the quality achieved at a pure Nash equilibrium in $\texttt{PNE}(\vecc s, \budget, T)$, with the quality of the optimal solution when constraint $T$ does not exist. To avoid confusion, for any instance, we define \texttt{qOPT} to be the value of optimal solution of the instance when the $T$ constraint is not present, and $\texttt{qOPT}_\texttt{T}$ the value of the optimal solution of the instance when the $T$ constraint exists. Obviously, $\texttt{qOPT}_\texttt{T} \leq \texttt{qOPT}$. Finally, in all of our comparisons that happen throughout the proof that follows, the pure Nash equilibrium is compared to the optimal solution produced by Lemma \ref{lem:greedy}.

\begin{theorem}\label{thm:PoS1}
The PoA of the modified proportional mechanism when there is a single proposal and for every $i \in V$, we have $q_i \in \{0, 1, \alpha\}$, is at most 4 (up to an additive factor of at most $6\cdot \alpha$). 
\end{theorem}
\begin{proof}
Consider a pure Nash equilibrium of quality $x$, and let $ i^*$ be the largest index such that $\sum_{i=1}^{i^\star} \alpha \cdot s_i \le \beta$ \footnote{In case $i^\star=0$, then Remark \ref{rem:zo} applies and we refer the reader to Proposition \ref{prop:PoAzo} for the respective approximation guarantees.}, if we do not take into account the constraint $T$. We begin the proof by considering the following special cases separately:
\begin{itemize}
    \item If $x = 0$, then exactly zero reviews are written at the pure Nash equilibrium. This implies that for every agent $i$, we have $s_i\geq \frac{\beta}{1+\sqrt{\alpha}}>\frac{\beta}{\alpha+\sqrt{\alpha}}$, meaning that it is not beneficial for them to write either a good or an excellent review. Suppose that $\texttt{qOPT}_\texttt{T}>1+\sqrt{\alpha}$ and that $E^g$ and $G^g$ are the sets of agents that write excellent and good reviews in this solution respectively. Then, this implies that:
    
   \begin{align*}
    &|E^g|\cdot \alpha+|G^g|>1+\sqrt{\alpha}\\
    \Rightarrow\quad
    & \frac{|E^g|\cdot \alpha}{1+\sqrt{\alpha}}+\frac{|G^g|}{1+\sqrt{\alpha}}>1\\
     \Rightarrow\quad
    & \frac{\beta \cdot |E^g|\cdot \alpha}{1+\sqrt{\alpha}}+\frac{\beta \cdot |G^g|}{1+\sqrt{\alpha}}>\beta \\
     \Rightarrow\quad
    &  \alpha \cdot \sum_{i\in E^g} s_i+ \sum_{i\in G^g} s_i>\beta,
\end{align*}
    a contradiction. Therefore, the value of $\texttt{qOPT}_\texttt{T}$ is at most $1+ \sqrt{\alpha}<6\cdot \alpha$, and thus the approximation guarantee is covered by the additive factor.
    \item If $0<i^\star < 6$, then \Cref{lemma:greedy} implies that $\texttt{qOPT}_\texttt{T} \le \alpha \cdot (i^\star + 1) \le 6\cdot \alpha$, which is covered by the additive factor.
\end{itemize}

For the rest of the proof, we will consider the case where $x \geq 1$ and $i^\star \ge 6$. Initially, let $0\leq d_{\frac{i^*}{4}}<1$ be the fractional part of $i^*/4$. Now assume that the quality of the pure Nash equilibrium is $x<\frac{\alpha \cdot i^\star}{4}-\alpha+\epsilon < \frac{\alpha \cdot i^\star}{4}-\sqrt{\alpha}$\footnote{Notice that for any $i^*\geq 6$ and $\alpha\geq 2$, we have $\frac{\alpha \cdot i^\star}{4}-\alpha+\epsilon\geq \frac{1}{2}\cdot \alpha+\epsilon>1$.}, for some arbitrarily small $0<\epsilon<<\alpha \cdot(1-d_{\frac{i^*}{4}})$, as otherwise, if
%Therefore, %if all agents $i' \le i^\star/4$ can write an excellent review (and are not restricted by $T$), 
%we have that: %\textcolor{red}{The statement here has to change, and probably is wrong as it is right now. What we have to do is to say: If $i^{\star}$ is greater than 5 then... if $i^{\star}$ is less than 5, then we know that at least one of the excellent review agent will write a an excellent review at the PNE, thus we once again get the additional factor of $4\cdot \alpha$. In that way the $5 \cdot \alpha$ that appears there is not needed.}
$$
x \ge \frac{\alpha \cdot i^\star}{4}-\alpha+\epsilon,
$$
then we know by \Cref{lemma:greedy}, that the optimal quality that can be achieved when $T$ does not exist, is at most $\alpha \cdot (i^\star + 1)$, leading to:
\begin{align*}
    \quad x  \ge \frac{\alpha \cdot i^\star}{4}-\alpha+\epsilon
&\Rightarrow x \ge \frac{\texttt{qOPT}}{4} - \frac{\alpha}{4}-\alpha+\epsilon \\
&\Rightarrow x \ge \frac{\texttt{qOPT}_\texttt{T}}{4} - \frac{5 \cdot \alpha}{4}+\epsilon\\
&\Rightarrow x \ge \frac{\texttt{qOPT}_\texttt{T}}{4} - 6 \cdot \alpha,
\end{align*}
which directly gives us the desired bound. So, now notice that if $x < \frac{\alpha \cdot i^\star}{4}-\alpha$, this implies that at least one agent $i \le \lfloor \frac{i^\star}{4}\rfloor$ does not write an excellent review at this pure Nash equilibrium. The reason for this is that otherwise we would have $x\geq \alpha \cdot \lfloor\frac{i^*}{4}\rfloor =  \alpha \cdot\frac{i^*}{4}-\alpha\cdot d_{\frac{i^*}{4}}=\frac{\alpha \cdot i^\star}{4}-\alpha+\alpha-\alpha\cdot d_{\frac{i^*}{4}}>\frac{\alpha \cdot i^\star}{4}-\alpha+\epsilon$,  a contradiction.  %In case all of them did, then the produced quality would be at least $\alpha\cdot \frac{i^\star}{4} > x$, contradicting our hypothesis.
So, say that $i'$ is the minimum indexed such agent. The latter means that every agent $i<i'$ writes an excellent review at the pure Nash equilibrium that we consider. By \Cref{lemma:greedy}, we know that 
\begin{align}
    \sum_{i=1}^{i^\star} \alpha \cdot s_i \le \beta 
    &\Rightarrow \sum_{i=\lfloor i^\star/4 \rfloor}^{i^\star} \alpha \cdot s_i \leq \beta \notag\\
    &\Rightarrow \alpha \cdot s_{\lfloor i^\star /4 \rfloor} \left(\frac{3}{4}\cdot i^\star\right) < \beta \notag\\
    &\Rightarrow s_{i'} < \frac{4}{3} \cdot \frac{\beta}{\alpha \cdot i^\star}\label{eq:hopt},
\end{align}\footnote{If $\frac{i^*}{4}<2$, then this is the only case where we have $\sum_{i=1}^{i^\star} \alpha \cdot s_i \leq \beta \Rightarrow \alpha \cdot s_{1} \left(\frac{3}{4}\cdot i^\star\right) < \beta$. In all other cases, $\sum_{i=\lfloor i^\star/4 \rfloor}^{i^\star} \alpha \cdot s_i < \beta$.}
where the last two inequalities hold as the $s_i$'s are positive and increasing in $i$. We will break the proof into cases that depend on how much agent $i'$ is affected by the constraint $T$.

\paragraph{Case 1:} $\alpha \cdot s_{i'} \leq T$. We split this case into two sub-cases. In the first, agent $i'$ does not write a review at the pure Nash equilibrium, while in the second, she writes a good review. In either case, we show that deviating to writing an excellent review increases agent $i'$'s utility, leading to a contradiction.
\paragraph{Case 1a:} Suppose that agent $i'$ does not write a review at the pure Nash equilibrium. It is easy to see that in this case, by deviating to writing a review of good quality, she would attain a utility of
\begin{align*}
\frac{\beta }{x+ \sqrt{\alpha}+1}-s_{i'} 
    \ge \frac{\beta }{\frac{\alpha \cdot i^\star}{4} + 1}-s_{i'}
    = \frac{4 \cdot \beta}{\alpha \cdot i^\star + 4}-s_{i'}.
\end{align*}
By (\ref{eq:hopt}), this deviation is beneficial if the following inequality holds:
\begin{align*}
\frac{4\cdot \beta}{\alpha \cdot i^\star + 4} \ge \frac{4}{3} \cdot \frac{\beta}{\alpha \cdot i^\star}\\ 
\Rightarrow
\frac{1}{\alpha \cdot i^\star + 4} \ge \frac{1}{3} \cdot \frac{1}{\alpha \cdot i^\star}\\
\Rightarrow
\alpha \cdot i^\star+ 4 \le 3\cdot \alpha \cdot i^\star\\
\Rightarrow
 2 \le  \alpha \cdot i^\star
\end{align*}
which is true for $\alpha \ge 2$ and $i^\star \ge 6$. Therefore, writing a good review would yield positive utility for agent $i'$, leading to a contradiction.

\paragraph{Case 1b:} Now, assume that agent $i'$ writes a good review at the pure Nash equilibrium. This implies that her utility by deviating to writing an excellent review would be lower:
\begin{equation}\label{eq:weird_function}
\frac{\beta}{x+\sqrt{\alpha}} - s_{i'} \geq \frac{\alpha \cdot \beta }{x+\sqrt{\alpha}-1+\alpha}-\alpha \cdot s_{i'}
\Rightarrow
(\alpha-1)\cdot s_{i'} \geq \beta \cdot\left(\frac{\alpha}{x+\sqrt{\alpha}-1+\alpha}-\frac{1}{x+\sqrt{\alpha}}\right).
\end{equation}
Let's call $f(x) = \frac{\alpha}{x+\sqrt{\alpha}-1+\alpha}-\frac{1}{x+\sqrt{\alpha}}$ the function appearing at the right hand side of inequality (\ref{eq:weird_function}), and set $y=x+\sqrt{\alpha}$. By studying the monotonicity of function for $\alpha \geq 2$, and $x \geq 1$, we have the following:
\begin{align*}
    &-\frac{\alpha}{(y-1+\alpha)^2}+\frac{1}{y^2}=0\\
    \Rightarrow\quad
    & (y-1+\alpha)^2=\alpha \cdot y^2\\
     \Rightarrow\quad
    & (\alpha-1) \cdot y^2- 2\cdot y \cdot (\alpha-1)- (\alpha-1)^2=0\\
     \Rightarrow\quad
    & y^2- 2\cdot y- (\alpha-1)=0
\end{align*}
For the solutions of this equation, we have that $y_{1}=1+\sqrt{\alpha}$ and $y_{2}=1-\sqrt{\alpha}$, and therefore, $x_{1}=1$ and $x_{2}=1-2\cdot \sqrt{\alpha}$. Finally, it is easy to see that for $x\geq x_1=1$ function $f(x)$ is decreasing. The latter guaranties that if $1\leq x\leq \frac{\alpha \cdot i^\star}{4}-\sqrt{\alpha}$, which is the interval that we are currently interested in, then $x=\frac{a\cdot i^\star}{4}-\sqrt{\alpha}$ is where $f(\cdot)$ is minimized. Going back to equation (\ref{eq:weird_function}), and using the monotonicity of function $f(\cdot)$ and equation (\ref{eq:hopt}), it should also hold that:
\begin{align*}
    &(\alpha-1)\cdot \frac{4}{3} \cdot \frac{\beta}{\alpha \cdot i^\star} 
    \geq 
    \beta \cdot  \left(\frac{\alpha}{x+\sqrt{\alpha}-1+\alpha}-\frac{1}{x+\sqrt{\alpha}}\right)\\
    \Rightarrow\qquad
    &(\alpha-1)\cdot \frac{4}{3} \cdot \frac{1}{\alpha \cdot i^\star} \geq \frac{\alpha}{\frac{\alpha \cdot i^\star}{4}-1+\alpha}-\frac{1}{\frac{a\cdot i^\star}{4}} = f\left(\frac{a\cdot i^\star}{4}-\sqrt{\alpha}\right),   
\end{align*}
substituting $x$ for $\alpha \cdot i^\star / 4-\sqrt{\alpha}$. We will show, that this is not possible, and since it is not possible for $x=\frac{a\cdot i^\star}{4}-\sqrt{\alpha}$, it cannot be possible for any $1\leq x\leq \frac{a\cdot i^\star}{4}-\sqrt{\alpha}$, as $f(\cdot)$ is decreasing in this interval. Simplifying the above inequality, we have the following:

\begin{align*}
    & \frac{1}{3} \cdot \frac{\alpha-1}{\alpha \cdot i^\star} \geq \frac{\alpha}{\alpha \cdot i^\star-4+4\cdot \alpha}-\frac{1}{\alpha\cdot i^\star}\\
    \Rightarrow\quad
    & \frac{\alpha+2}{ 3\cdot i^\star} \geq \frac{\alpha}{\alpha \cdot i^\star-4+4\cdot \alpha}\\ 
    \Rightarrow\quad
    & 3\cdot \alpha^2 \cdot i^\star \leq (\alpha+2)(\alpha \cdot i^\star-4+4\cdot \alpha) \\
    \Rightarrow\quad
    & 3\cdot \alpha^2 \cdot i^\star \leq  \alpha^2 \cdot i^\star-4\cdot \alpha+4\cdot \alpha^2+2\cdot \alpha \cdot i^\star-8+8\cdot \alpha  \\
    \Rightarrow\quad
    & 2\cdot \alpha^2 \cdot i^\star \leq 4\cdot \alpha^2+2\cdot \alpha \cdot i^\star-8+4\cdot \alpha  \\
    \Rightarrow\quad
     & \alpha^2 \cdot i^\star \leq 2\cdot \alpha^2+\alpha \cdot i^\star-4+2\cdot \alpha  \\
    \Rightarrow\quad
    & i^*(\alpha^2-\alpha)+4 \leq 2\cdot (\alpha^2+ \alpha),
\end{align*}
which never holds when $i^* \geq 6$ and $\alpha \geq 2$, therefore we end up to a contradiction.

%Note that this inequality cannot hold for any $\alpha \ge 2$ and $i^\star \ge 6$. Additionally, by (\ref{eq:maximum}), (\ref{eq:at2}) and the previous monotonicity observations, this inequality cannot hold for any $2 \le  x \le \alpha \cdot i^\star / 4$. Concluding, as long as $x \ge 2$ and $i^\star \ge 6$ (i.e. at least 6 excellent reviews are written in the optimal outcome), agent $i'$ would have a profitable deviation by writing an excellent review, leading to a contradiction. %For the remaining case of $x = 1$ we can only have exactly one good review, but since $i^\star \ge 6$ we know that at least 6 excellent reviews are written in the optimal outcome. Clearly, the $s_i$'s of at least two of them need to be no greater than $\beta / 2$. But in this case, in our PNE there would have been an additional agent having a profitable deviation and writing a good review.\textcolor{red}{The last sentences are not very formal.}

\paragraph{Case 2:} $s_{i'}\leq T<\alpha \cdot s_{i'}$. We split the agents into $E$ and $G$, based on if they write an excellent or good review respectively at $\texttt{qOPT}_\texttt{T}$. Notice set $E$ contains only agents $i<i'$. 
\paragraph{Case 2a:} $\alpha \cdot |E| \geq \frac{\texttt{qOPT}_\texttt{T}}{2}$. By the definition of $i'$ we know that every $i<i'$ writes an excellent review at the pure Nash equilibrium that we consider. The latter implies a implies a half-approximation of $\texttt{qOPT}_\texttt{T}$. In case $i'=1$, then again, by the definition of $i'$ we get that $E=\emptyset$. Due to the case that currently consider, this means that the value of $\texttt{qOPT}_\texttt{T}$ is $0$. %Therefore we have that every $i \in E$ writes an excellent review at the pure Nash equilibrium   By following the same arguments as in Case 1, it is easy to see that every $i \in E$ writes an excellent review at the pure Nash equilibrium that we consider (since these agents are no longer constrained by $T$), something that implies a half-approximation of $\texttt{OPT}_\texttt{T}$. 
\paragraph{Case 2b:} $|G| \geq \frac{\texttt{qOPT}_\texttt{T}}{2}$. Assume that at the pure Nash equilibrium, a total quality of $x<\frac{|G|}{2}$ is achieved. Let $G_1 \subseteq G$ to contain half of the agents in $G$, and in particular the ones with the smallest indices (thus, the ones with the smallest $s_i$'s). Notice that at least some $i'' \in G_1$ does not write either a good review or an excellent at the pure Nash equilibrium, as otherwise we would have that $x\geq \frac{|G|}{2}$. Let $i^l = \max_i G_1$. Since $\sum_{i \in E}\alpha \cdot s_i+ \sum_{i \in G}s_i \leq \beta $, it is easy to see that 

\begin{equation}\label{eq:poutsa}
s_{i''}\leq s_{i^l} < \frac{2\cdot(\beta - \sum_{i \in E} \alpha \cdot s_i)}{|G|},
\end{equation}
as otherwise the sum of the skills of the agents in $G \setminus G_1$, would exceed the available budget. So the only thing that remains to show, is that agent $i''$ is better off by writing a good review, something that will lead to contradiction. The latter, is true as
\begin{align}
u_{i''}=  \frac{\beta}{x+\sqrt{\alpha}+1} -s_{i''} 
&\geq \frac{\beta}{\frac{\texttt{qOPT}}{4}-\sqrt{ \alpha} +\sqrt{\alpha}} -s_{i''} \\ 
&\geq \frac{2 \cdot \beta}{|G|}- s_{i''}> 0,
\end{align}
where last inequality holds from Equation (\ref{eq:poutsa}).

\paragraph{Case 3: $T< s_{i'}$:} Once again, we split the agents into $E$ and $G$, based on if they write an excellent or good review respectively at $\texttt{qOPT}_\texttt{T}$.  Now notice that no agent $i\geq i'$ writes a review at $\texttt{qOPT}_\texttt{T}$ due to constraint $T$, while by the definition of $i'$, every agent $i<i'$ writes an excellent review at the pure Nash equilibrium that we consider. The latter implies an 1-approximation of $\texttt{qOPT}_\texttt{T}$. In case $i'=1$, then for every agent $i \neq i'$ we have that $s_i \geq s_i^{'}$. Therefore, we get that the value of $\texttt{qOPT}_\texttt{T}$ is $0$. %As with Case 2a, it is easy to see that since for every $i \in E$ we have that $\alpha \cdot s_i \leq T $, and for every $i \in G$ we have that $ s_i \leq T $, we can derive that for every $i$ in both groups, we have $s_i < s_{i'}$.  By following the same arguments as in Case 1, it is easy to see that every  $i \in E$ writes an excellent review, and every $i \in G$ writes a good review  at the pure Nash equilibrium that we consider, something that implies an 1-approximation of $\texttt{OPT}_\texttt{T}$. 

%The statement follows by combining all the aforementioned cases.
\end{proof}

% Theorems \ref{thm:pne} and \ref{thm:PoS1} imply the following result:

% \begin{corollary}
% When there is only a single proposal, and for every $i\in V$ we have $q_i \in \{0, 1\}$, the proportional mechanism guarantees a tight PoA of 2.
% \end{corollary}

% \begin{proof}
% The tightness can be seen directly by considering an instance of 2 agents with $s_1=\epsilon$, and $s_2=1-\epsilon$, under $B=2$, and $T=1$.
% \end{proof}

\begin{remark}
\normalfont Notice that the additive factor that appears in the approximation guarantee of the pure Nash equilibrium, captures some corner cases where the appearance of parameter $\alpha$ cannot be avoided. An easy example that demonstrates this is the following: Consider an instance where there is only one agent with skill $s_i=\frac{1}{\sqrt{\alpha}}$, $B=1$, and $T=1+\epsilon$ for some $\epsilon>0$ and $\alpha \geq 4$. The utility of this agent when she writes a good review is $\frac{1}{1+\sqrt{\alpha}}-\frac{1}{\alpha}>0$, which is more than $\frac{\alpha}{\alpha+\sqrt{\alpha}}-\frac{\alpha}{\alpha}<0$, her utility when she writes an excellent review. The first state describes a pure Nash equilibrium of quality 1, while the second state describes the optimal solution of quality $\alpha$. Something that is also interesting is that although it seems initially that is a side effect of the modification (the fact that $\sqrt{\alpha}$ is added to the denominator), it also cannot be avoided in the original proportional mechanism \cite{BirmpasKLO22}.
\end{remark}

For the remainder of the paper we revert back to the study of the original version of the proportional mechanism, as for the scenarios that will be studied from now on, the modification is no longer needed to guarantee equilibrium existence. 

\subsection{The Special Case of 0 and 1 Qualities}

We would like to point out that if one considers the special version of this case of the problem, where the available strategies for the agents are writing reviews of qualities 0 and 1 \footnote{A more general version of this case, where there are more that one proposals, is presented in the next section.}, then the existence of a pure Nash equilibrium under the original version of the proportional mechanism is always guaranteed, and one can get a tight PoA bound of 2. Regarding the existence of pure Nash equilibria for this case, it can directly be derived by the first paragraph of the proof of Theorem \ref{thm:pne} \footnote{For an alternative proof, we refer the reader to Theorem \ref{thm:pot}.}. It is also clear that the optimal quality in this version of the problem can be achieved by letting agents $\{1,\ldots, i^\mu\}$ to write a good review, where $i^\mu$ is the largest index for which $\sum_{i=1}^{i^\mu}s_i\leq\beta$, and $s_i^\mu \leq T$. The latter implies that the quality of optimal outcome is $i^\mu$.

\begin{proposition}\label{prop:PoAzo}
The PoA of the proportional mechanism when there is a single proposal and for every $i \in V$, we have $q_i \in \{0, 1\}$, is 2 (up to an additive factor of 1). Moreover, this is tight.
\end{proposition}

\begin{proof}
Consider a pure Nash equilibrium and suppose for contradiction that its quality is $x <\frac{i^\mu}{2}$. We will split the proof in three cases:
\begin{itemize}
    \item If $\frac{i^\mu}{2}$ is an integer, then this implies that at least one agent $i \le \frac{i^\mu}{2}$ does not write a review at this pure Nash equilibrium. Say that $i'$ is the minimum indexed such agent. From the previous discussion, we know that 
\begin{align}
    \sum_{i=1}^{i^\mu} s_i \le \beta 
    &\Rightarrow \sum_{i=i^\mu/2}^{i^\mu} s_i < \beta \notag\\
    &\Rightarrow s_{i^\mu /2} \cdot \frac{i^\mu}{2} < \beta \notag\\
    &\Rightarrow s_{i'} < \frac{\beta}{\frac{i^\mu}{2}}\label{eq:hopt2},
\end{align}
where the last two inequalities hold as the $s_i$'s are positive and increasing in $i$.

Since agent $i'$ does not write a review at the pure Nash equilibrium, it is easy to see that in this case by deviating to writing a review of good quality, she would attain a utility of
\begin{align*}
\frac{\beta }{x+1}-s_{i'} 
    \ge \frac{\beta }{\frac{i^\mu}{2}}-s_{i'}.
\end{align*}
By (\ref{eq:hopt2}), this deviation is always beneficial. %if the following inequality holds:
%$$
%\frac{2\cdot \beta}{i^\mu} \ge 2\cdot \frac{\beta}{ i^\mu} 
%\Rightarrow
%\frac{1}{i^\mu} \ge \frac{1}{i^\mu},
%$$
%which is true for any $i^\mu$ 
Therefore, writing a good review would yield positive utility for agent $i'$, leading to a contradiction.

\item If $\frac{i^\mu}{2}\geq 1$ and it is not an integer, then this implies that at least one agent $i \le \lfloor\frac{i^\mu}{2}\rfloor+1$ does not write a review at this pure Nash equilibrium. Say that $i'$ is the minimum indexed such agent. From the previous discussion, we know that 
\begin{align}
    \sum_{i=1}^{i^\mu} s_i \le \beta 
    &\Rightarrow \sum_{i=\lfloor i^\mu/2 \rfloor +1}^{i^\mu} s_i < \beta \notag\\
    &\Rightarrow s_{\lfloor i^\mu /2 \rfloor+1} \cdot (\lfloor\frac{i^\mu}{2}\rfloor+1) < \beta \notag\\
    &\Rightarrow s_{i'} < \frac{\beta}{\lfloor\frac{i^\mu}{2}\rfloor+1}\label{eq:hopt3},
\end{align}
where the last two inequalities hold as the $s_i$'s are positive and increasing in $i$.

Since agent $i'$ does not write a review at the pure Nash equilibrium, it is easy to see that in this case that by deviating to writing a review of good quality, she would attain a utility of
\begin{align*}
\frac{\beta }{x+1}-s_{i'} 
    \ge \frac{\beta }{\lfloor\frac{i^\mu}{2}\rfloor+1}-s_{i'}.
\end{align*}
By (\ref{eq:hopt3}), this deviation is always beneficial. %if the following inequality holds:
%$$
%\frac{2\cdot \beta}{i^\mu} \ge 2\cdot \frac{\beta}{ i^\mu} 
%\Rightarrow
%\frac{1}{i^\mu} \ge \frac{1}{i^\mu},
%$$
%which is true for any $i^\mu$ 
Therefore, writing a good review would yield positive utility for agent $i'$, leading to a contradiction.

\item If $\frac{i^\mu}{2}<1$ and is it is not an integer, then this implies that $i^\mu=1$ and that the optimal solution provides a total quality of 1. It is straightforward to see that any pure Nash equilibrium provides a total quality of either 1 or 0. The latter case captures the scenario where an agent is indifferent between writing a review or not, as in both cases her utility is 0. Although this implies an infinite PoA if we measure the performance of the equilibrium according to the definition, as this is the only case that something like that can happen, we view it as having an additive loss of 1.

\end{itemize}

Regarding tightness, consider an instance with 2 agents, skills $s_1=0.4$, and $s_2=0.6$, and finally $B=T=1$. It is easy to see that both of them participate in the optimal outcome, which has a total quality of 2, while only one of the can write a review at a pure Nash equilibrium, since otherwise agent 2 ends up with a negative utility as her reward is $\frac{1}{2}$ in that case. This implies a quality of 1 in any pure Nash equilibrium.
\end{proof}

\section{Multiple Proposals: The Coverage Objective}\label{sec:mprop}
%%%%%%%%%%%%%%%%%%%%%%%%%%%%%%%%%%%%%%%%%%%%%%%%%%%%%%%%

In this section we go beyond the case of the single proposal, and we turn our attention on the performance of the proportional mechanism with respect to the coverage objective. Specifically, we are interested in the number of proposals that end up with at least one review at an equilibrium state of the proportional mechanism, in comparison with the respective optimal solution. Since our priority is to maximize the number of proposals that are covered  with at least one review, the evaluation of the quality of these reviews takes the back seat. This leads to a version of the problem where $q_{ij} \in \{0,1\}$, i.e., an agent $i$, either writes a review for a proposal $j$, or she doesn't, while for the effort function $f(\cdot)$ we have,

\[	f(q_i)=  
	\begin{cases}
	0, & \text{if $q_i=0$} \\
	1, & \text{if $q_i=1$} 
	\end{cases} 
\]
for every $i \in V$. We begin by exploring the existence of pure Nash equilibria under this setting.

%%%%%%%%%%%%%%%%%%%%%%%%%%%%%%%%%%%%%%%%%%%%%%%%%%%
\subsection{Multiple Proposals: Existence of Pure Nash Equilibria}
%%%%%%%%%%%%%%%%%%%%%%%%%%%%%%%%%%%%%%%%%%%%%%%%%%%

We begin with the following theorem.

\begin{theorem}\label{thm:pot}
When there are multiple proposals and $q_{ij} \in \{0, 1\}$ for every $i, j$, then the proportional mechanism always has a pure Nash equilibrium.
\end{theorem}
\begin{proof}
 Our goal is to prove that the procedure of writing reviews under the proportional reward scheme, can be seen as a game that admits a \emph{potential} function. We start by associating each of the proposals $P_j$, with a function $\Phi_j(\cdot): \vecc q^j \rightarrow \mathbb{R}$. More specifically, we define $\Phi_j(\cdot)$ as follows:

$$
\Phi_j(\vecc q^j)= \beta \cdot H_{k}-\sum_{i \in K} s_{ij},
$$
where $K\subseteq V$, with $|K|=k$, is the set of agents for which we have $q^j_i=1$ in $\vecc q^j$, and $H_{k}$ is the \emph{$k$-th Harmonic number} defined as $H_{k} = 1 + \frac{1}{2} + \dots + \frac{1}{k}$.

Now consider the strategy vectors $\vecc q^j= (q^j_i,q^j_{-i})$, and $\vecc {\hat{q}}^j= (\hat{q}^j_i,q^j_{-i})$, where agent $i$ changes her strategy from $q^j_i$ to $\hat{q}^j_i$, while the strategies of the rest of the agents remain the same. We have that:

\[	\Phi_j(q^j_i,q^j_{-i})-\Phi_j(\hat{q}^j_i,q^j_{-i})=  
	\begin{cases}
	\frac{\beta}{k}-s_{ij}, & \text{if $q^j_i=1$, $\hat{q}^j_i=0$} \\
	s_{ij}-\frac{\beta}{k+1}, & \text{if $q^j_i=0$, $\hat{q}^j_i=1$}  \\
	0, & \text{if $q^j_i=\hat{q}^j_i$}
	\end{cases} 
\]
 which implies that
 \begin{equation*}
 \Phi_j(q^j_i,q^j_{-i})-\Phi_j(\hat{q}^j_i,q^j_{-i})=u_i(q^j_i,q^j_{-i})-u_i(\hat{q}^j_i,q^j_{-i}).
 \end{equation*}
 Observe that this holds for every $P_j$ and $\Phi_j$. 
 
 We proceed by defining function $\Phi(\vecc q)=\sum_{j=1}^m\Phi_j(\vecc q^j)$. Notice that for every reviewer $i$, we have that: 
 
 \begin{align*}
  \Phi(\vecc q_i, \vecc  q_{-i})-\Phi(\vecc{\hat{q}}_i, \vecc q_{-i})
  &= \sum_{j=1}^m [\Phi_j(q^j_i,q^j_{-i})-\Phi_j(\hat{q}^j_i,q^j_{-i})] \\
  &= \sum_{j=1} ^m[u_i(q^j_i,q^j_{-i})-u_i(\hat{q}^j_i,q^j_{-i})] \\
  &= u_i(\vecc q_i, \vecc  q_{-i})-u_i(\vecc{\hat{q}}_i, \vecc q_{-i}).
 \end{align*}
 Therefore, function $\Phi(\cdot)$ is an exact potential function and thus, there always exists at least one pure Nash equilibrium \cite{Sap1996}.
\end{proof}

%%%%%%%%%%%%%%%%%%%%%%%%%%%%%%%%%%%%%%%%%%%%
\subsection{Coverage Approximation Guarantees}
%%%%%%%%%%%%%%%%%%%%%%%%%%%%%%%%%%%%%%%%%%%%

In this section, we explore the performance of the proportional mechanism with respect to the coverage objective. Unfortunately, as the following proposition shows, the guarantees that it provides at a pure Nash equilibrium can be as bad $n$-approximate to the optimal coverage outcome.

\begin{proposition}\label{lem:covapprx}
The PoA obtained by the proportional mechanism  can be as bad as $n$.
\end{proposition}
\begin{proof}
 Suppose that we have an instance with $n$ agents and $n$ proposals, where the available budget is $B=n$, and the maximum effort constraint is $T=1$. In addition, let every agent $i$ to have a very small skill parameter for the first proposal, i.e., $s_{i1}=\epsilon>0$, and a skill parameter of 1 for the rest of the proposals, i.e., $s_{ij}=1$ for $P_j \in P \setminus P_1$.

Initially, notice that in any feasible $\vecc q$, each agent writes at most one review, as otherwise the maximum effort constraint is violated.  Thus, the optimal solution with respect to the coverage objective is produced when each agent writes one review, and at the same time each proposal has exactly one review. The aforementioned family of assignments guarantees a coverage of $n$.

Let us now turn our attention to the proportional mechanism. Recall that the available budget is distributed equally among the proposals, so for the described instance we have that $\beta=1$ for any $P_j \in P$. Consider the case where every agent writes a review for the first proposal. The derived utility of every agent $i$ is $u_i=\frac{1}{n}-\epsilon$. It is easy to see that this is a pure Nash equilibrium, as every agent can write at most one review, and if any agent tries to deviate to either not writing a review, or writing a review for a different proposal, then her utility becomes zero. Notice that for this specific example, the described assignment is the only possible pure Nash equilibrium, as if there are less than $n$ reviews for the first proposal, this implies that there is an agent $i$ that has not written a review for it. Thus, by playing $q_{i1}=1$ and $q_{ij}=0$ for $j\neq 1$, she achieves a higher utility. Therefore, there is only one pure Nash equilibrium in this instance and the coverage guarantee that it provides is $1$.
\end{proof}

As the result of Proposition \ref{lem:covapprx} is negative, the next natural direction would be to explore the performance of the proportional mechanism under augmented resources. %The intuition behind this is that if the mechanism could provide higher rewards to the agents, then this could possibly motivate some of them to write reviews for proposals which, under the original budget, do not get any review. 
Our goal, is to investigate whether by increasing the budget, the proportional mechanism can achieve at its equilibrium states, a total coverage that is identical to the coverage of the optimal (under the original budget) solution. %We give the answer with the following lemma.

\begin{proposition}\label{lem:bigB}
There are instances where the the optimal coverage cannot be achieved at a pure Nash equilibrium, no matter the increase in the budget.
\end{proposition}
\begin{proof}
Suppose that we have 2 agents and 3 proposals. The skill parameters of both the agents are as follows: $s_{i1}=s_{i2}=\epsilon$, for some $0<\epsilon<\frac{1}{2}$, and $s_{i3}=1$, for $i \in \{1, 2\}$. Moreover, let the total budget be $B=3 \cdot \beta$, for some $\beta\geq 1$, and notice that this implies that the $\beta$ is the available reward for every $j \in \{1,2,3\}$). Finally let $T=1$. It is easy to see that in the optimal coverage solution, every proposal ends up with one review, i.e., one of the agents writes a review for the first two proposals, while the other writes for the third one. Thus, we achieve a total coverage of 3.

Now consider an outcome where both agents write a one review for each of the first two proposals. Thus, their utilities are $u_i=2\cdot(\frac{\beta}{2}- \epsilon)=\beta-2\cdot \epsilon>0$, for every $i \in \{1, 2\}$. It is easy to confirm that this is a pure Nash equilibrium. Since the utility of both agents is positive in this allocation, the only possible options for both them are either to write a review for one of the first two proposals (something that leads to a lower utility), or to write just one review for the third proposal (as constraint $T$ dictates). The latter deviation provides an agent with a utility of $\beta-1=\beta - 2\cdot \frac{1}{2}<\beta-2\cdot \epsilon$.

By the above discussion we get that the aforementioned pure Nash equilibrium, only two out of three proposals are covered, something that implies a \texttt{PoA}$=\frac{3}{2}$ \footnote{Notice that this is actually the only pure Nash equilibrium of this instance.}. The statement follows from the fact that this holds for any $\beta \geq 1$, so the overall available budget does not affect the coverage guarantee that a pure Nash equilibrium can achieve.
\end{proof}

We conclude this section with a positive result. Even though Proposition \ref{lem:bigB} demonstrates that the optimal coverage cannot be achieved at a pure Nash equilibrium regardless of the increase in the budget, we show that doubling the budget is enough for a PoA of 3 to be guaranteed.

\begin{theorem}
The proportional mechanism guarantees a {\normalfont $\texttt{PoA}_2=3$}. Moreover, this result is tight.
\end{theorem}

\begin{proof}
Let $\texttt{cOPT(B)}$ to be the optimal coverage solution of an instance under budget $B$. Assume that at $\texttt{cOPT(B)}$, a set of proposals $K \subseteq P$ is covered with at least one review, and let $\vecc q$ to be a strategy vector that defines a pure Nash equilibrium under budget $2\cdot B $. Let $S \subseteq P$, to be the set of proposals that are covered at $\texttt{cOPT(B)}$ and not at the pure Nash equilibrium defined by $\vecc q$. In addition, define $D \subseteq V$  as the set of agents that at $\texttt{cOPT(B)}$, write at least one review for a non-empty set of proposals in $S$. Finally, let $A^{opt}_i \subseteq S$  be the set of proposals for which an agent $i \in D$ writes a review in $\texttt{cOPT(B)}$. 

Now for the remainder of the proof, it is crucial to define set $D^\star \subseteq D$, which will be a sufficiently small subset of agents, the reviews of which in the optimal solution $\texttt{cOPT(B)}$, cover every proposal in $S$ with at least one review.  We build set $D^\star$ according to the following greedy procedure.  

\begin{algorithm}[ht]
		\caption{Greedy Set Construction $(D, S)$}
		\begin{algorithmic}[1]
			\State $D^\star=\emptyset$\textbf{;} $S^\star=\emptyset$
			\While{$S^\star \neq S$} 
			\State $h = \arg\max_{i \in D\setminus D^\star} |A^{opt}_i \cap (S\setminus S^\star)|$ \Comment{Find the agent from $D$, that covers the highest number of proposals in $S$ that are currently uncovered. Break ties arbitrarily.}
			\State $S^\star=S^\star \cup A^{opt}_h$ 
			\State $D^\star =D^\star \cup \{h\}$ \label{line:rr6} \vspace{-2pt}
			\EndWhile
			\State 
			\Return {$D^\star$}
		\end{algorithmic}
		\label{alg:GP}
	\end{algorithm} 
	
	 Our main goal is to find a subset of agents in $D$, the cardinality of which is at most $|S|$, and by considering the reviews that this subset writes in the optimal solution, every proposal in $S$ can be covered with at least on review. To this end, we start by defining set $S^\star$ as the set of the currently covered proposals (which is empty in the beginning). We then proceed by considering set $D^\star$ which is also initially empty. At each step we add to it the agent from set $D$ that covers (according to the optimal solution) the highest number of proposals in $S$ that are currently uncovered (we break the ties arbitrarily), and we update set $S^\star$ accordingly. The procedure stops when $S^\star$ becomes equal to $S$. It is easy to see that this will happen after at most $|S|$ steps, as otherwise this would mean that there is a step $j\leq |S|$ where an agent that covered zero non-covered proposals was added to set $D^\star$. However, since at each step we add the agent that covers the highest number of uncovered proposals, this would imply that at every step $j'>j$, the agent that is added to set $D^\star$ does not contribute to the coverage of set $S$, and thus that set $D$ cannot cover set $S$, a contradiction. Therefore, since at each step, an agent is added, and the procedure takes at most $|S|$ steps, we get that $|D^\star| \leq |S|$.

 We proceed by exploring some structural properties of the the pure Nash equilibrium $\vecc{q}$. Initially, notice that every agent $i \in D^\star$ writes at least one review at the pure Nash equilibrium $\vecc{q}$, as otherwise, by writing reviews for the proposals in $A^{opt}_i$ she would end up with a utility of, 

\begin{align*}
u_i&=\sum_{j \in A^{opt}_i} 2\cdot \beta - \sum_{j \in A^{opt}_i} s_{ij}  \geq \sum_{j \in A^{opt}_i} 2\cdot \beta - \sum_{j \in A^{opt}_i} \beta  >0,
\end{align*}
where the first inequality holds due to the fact that for each $j \in A^{opt}_i$, we have that $s_{ij}\leq \beta$, as otherwise $\texttt{cOPT(B)}$ would not be budget feasible. 

From the previous argument, we get that every agent $i  \in D^\star$ in the pure Nash equilibrium defined by $\vecc q$, writes at least one review for a subset of proposals in $K\setminus S$. For every $i \in D^\star$, name this subset of proposals $A^{pne}_i$. Finally, let $n_j$ to be the number of reviews that a proposal $j \in K\setminus S$ gathers at the pure Nash equilibrium. For every $i \in D^\star$, the definition of the pure Nash equilibrium gives us the following:

\begin{equation} \label{eq:pnecov}
    \sum_{j \in A^{pne}_i}\left(\frac{2 \cdot \beta}{n_j}-s_{ij}\right) \geq \sum_{j \in A^{opt}_i}\left(2 \cdot \beta -s_{ij}\right).
\end{equation}
At the same time, we have
 
\begin{equation} \label{eq:pneskill}
    \sum_{j \in A^{pne}_i}s_{ij} > \sum_{j \in A^{opt}_i}s_{ij}-\min_{j \in A^{opt}_i}s_{ij},
\end{equation} 
as otherwise, agent $i$ would be able to write a review for every $ j \in A^{pne}_i$, plus a review for $\argmin_{j \in A^{opt}_i}s_{ij}$, without violating $T$. This however would improve utility, something that contradicts the pure Nash equilibrium assumption.

By combining equations (\ref{eq:pnecov}) and (\ref{eq:pneskill}), we get 

\begin{align*}
    &\sum_{j \in A^{pne}_i}\frac{2 \cdot \beta}{n_j}-\min_{j \in A^{opt}_i}s_{ij} > \sum_{j \in A^{opt}_i}2 \cdot \beta\\
    \Rightarrow\quad 
    &\sum_{j \in A^{pne}_i}\frac{2 \cdot \beta}{n_j} > \sum_{j \in A^{opt}_i}2 \cdot \beta  -\min_{j \in A^{opt}_i}s_{ij}.
\end{align*}
Since this holds for every agent $i \in D^\star$, we have 
\begin{align*}
   \sum_{i \in D^\star}\sum_{j \in A^{pne}_i}\frac{2 \cdot \beta}{n_j} > \sum_{i \in D^\star}\sum_{j \in A^{opt}_i}2 \cdot \beta  -\sum_{i \in D^\star}\min_{j \in A^{opt}_i}s_{ij}.
\end{align*}
Now notice that,
\begin{align*}
\sum_{i \in D^\star}\sum_{j \in A^{pne}_i}\frac{2 \cdot \beta}{n_j}\leq \sum_{i \in V}\sum_{j \in A^{pne}_i}\frac{2 \cdot \beta}{n_j}=2\cdot \beta \cdot |K\setminus S |, 
\end{align*}
 where the first inequality holds because we sum over all agents in $V \supseteq D^\star$, while the the sum in the right hand side of the inequality also represents the sum of the payments that the agents derive for each proposal, which is equal to $2 \cdot \beta$ (per proposal). Additionally,

\begin{align*}
\sum_{i \in D^\star}\sum_{j \in A^{opt}_i}2 \cdot \beta \geq 2\cdot \beta \cdot |S|, 
\end{align*}
as we know that the agents in $D^\star$, cover every proposal in $S$ with at least one review in $\texttt{cOPT(B)}$. Therefore, we derive that,

\begin{align*}
 &\qquad2\cdot \beta \cdot |K \setminus S|>2\cdot \beta \cdot |S|-\sum_{i \in D^\star}\min_{j \in A^{opt}_i}s_{ij}\\ 
 &\Rightarrow \sum_{i \in D^\star}\min_{j \in A^{opt}_i}s_{ij} > 2\cdot \beta \cdot (|S|-|K\setminus S|) \\
 &\Rightarrow \sum_{i \in D^\star}\min_{j \in A^{opt}_i}s_{ij} >  2\cdot \beta \cdot (2 \cdot |S|-|K|).
\end{align*}
On the other hand, we have that

\begin{align*}
  \sum_{i \in D^\star}\min_{j \in A^{opt}_i}s_{ij}&\leq |D^\star|\cdot \max_{i \in D^\star}\min_{j \in A^{opt}_i}s_{ij} \\ & \leq |S|\cdot  \max_{i \in D^\star}\min_{j \in A^{opt}_i}s_{ij} \\ & \leq |S| \cdot \beta.
\end{align*}
The statement follows from the fact that if

%\textcolor{red}{Add here why $|D|$ is smaller than $|S|$. Update: Tbh I do not think that this is correct, $|D|$ can be larger than $|S|$ (imagine for example that $|S|=1$, there is a chance that more than one agents write a review for this proposal). I think that $D$ should be defined as any set with cardinality smaller or equal to $S$, that covers $S$ in the optimal solution. Such a set always exists. This needs proof but I think that it is easy to show. Consider $D$ defined as it is right now, and $D^\star \subseteq D$ defined as follows: Start with an agent in $D$ that covers the highest number of proposals in $S$ (break ties arbitrarily) and each time add the agent that covers the highest number of proposals that are not yet covered. Notice that after at most $|S|$ such inclusions, set $S$ will be covered as otherwise one agent that is included in set $D^\star$ in the first $|S|$ steps, did not cover any new proposal in $S$, and since the procedure is greedy, this means that the same will hold for any agent that will be included in the subsequent steps, a contradiction as this implies that $D$ does not cover $S$. Therefore, $|D^\star|\leq |S|$. If we use set $D^\star$ instead of $D$, then our proof is correct. By the way, I think that we can prove it even if we just consider $D$, but some of the inequalities have to be changed.} The statement follows from the fact that if 

\begin{equation*}
2\cdot \beta \cdot (2 \cdot |S|-|K|) > |S| \cdot \beta \iff |S|> \frac{2|K|}{3},
\end{equation*}
then we have a contradiction. 

Regarding tightness, consider the following example: there are $3\cdot k-1$ agents and $3\cdot k-1$ proposals, where $k\geq 2$. The total budget is $3 \cdot k-1$, thus $\beta =1$ for each proposal, and finally $T=1$. The agents are split in sets $L_1$ and $L_2$, with cardinalities $2 \cdot k-1$ and  $k$ respectively, while the proposals are split in sets  $F_1$ and $F_2$, with cardinalities $k$ and $2 \cdot k-1$ respectively. Every agent $i \in L_1$ has $s_{ij}=\epsilon$ for $j\in F_1$, and $s_{ij}=1$ for $j\in F_2$. On the other hand, every agent $i \in L_2$ has $s_{ij}=1$ for $j\in F_1$, and $s_{ij}=3$ for $j\in F_2$. Under this configuration, it is clear that the agents of $L_2$ cannot write more than one reviews, as $T=1$. Now it is easy to see that at $\texttt{cOPT(B)}$ every agent in $L_1$ writes one review for just one proposal in $F_2$ so that every proposal in $F_2$ has exactly one review, while every agent in $L_2$ writes one review for just one proposal in $F_1$ so that every proposal in $F_1$, once again has exactly one review. Therefore, every proposal is covered in the optimal solution.

Now consider the strategy vector $\vecc q$ (under budget $2\cdot B$) that defines the following allocation: every agent in $L_1$ writes one review for every proposal in $F_1$, while every agent in $L_2$ does not write a review. Initially, notice that for the agents of $L_2$. there is no profitable deviation as they cannot write a review for the proposals in $F_2$, and if they try to write a review for a proposal in $F_1$, their utility is $\frac{2}{2\cdot k}-1<0$. As for the agents of $L_1$, their only meaningful deviation is to write one review for a proposal in $F_2$, but this provides them with a utility of $2-1<\frac{2\cdot k}{2\cdot k -1}-k\cdot \epsilon =u_i$, where $u_i$ is their current one. Thus, this is a pure Nash equilibrium under budget $2\cdot B$ \footnote{It is easy to confirm that this is also a pure Nash equilibrium under budget $B$.}, that is also $\frac{1}{3}$-approximate to $\texttt{cOPT(B)}$ \footnote{Notice that if we further increase the budget, this is no longer a pure Nash equilibrium.}.
\end{proof}

%%%%%%%%%%%%%%%%%%%%%%%%%%%%%%%%%%%%%%%
\section{Performance with Real Data}\label{sec:RD}
%%%%%%%%%%%%%%%%%%%%%%%%%%%%%%%%%%%%%%%

A commercially important application of the previous model is \emph{Project Catalyst}, which was outlined in the introduction of this work (and additional information can be found \href{https://projectcatalyst.io/}{here}). Following our model, the community advisors (CA) who write reviews for the proposals take up the role of the reviewers. Each of them can write as many reviews as they like, for any community generated proposal they are prefer. These reviews are then reviewed again by the veteran community advisors (vCA) and are assigned a grade that can be either `excellent', `good' or `filtered out', for reviews found to be below the minimum quality threshold. The CA's know the public guidelines that define the requirements of the different qualities and should, at the time of submission, have a good estimate of the grade that their review could get. According the terminology of our model, the parameter $\alpha = 3$ is used to increase the rewards given to excellent reviews (i.e. a `good' review contributes 1 to the quality and an excellent contributes 3). Following the reviewing phase, the voters can cast a `Yes', `No' or `Abstain' vote for each proposal, with every vote having weight proportional to the users' stake in ADA, the currency used by Cardano. More details about the voting process (and it's cryptographic guarantees) can be found in~\cite{treasury_paper}. 

We present the data from Fund7, the latest completed (as of May 2022) round of funding. The data set is publicly available \href{https://docs.google.com/spreadsheets/d/1ZM3ytXkMB34iSo2LamNxpver-rs9fShnpeNEia-VdBo/edit?usp=sharing}{here}.  From the total amount of \$8,000,000:
\begin{itemize}
    \item \$6,400,000 were rewarded to funded proposals.
    \item \$320,000, spread across the $712$ candidate proposals, were used to pay for CA reviews.
    \item \$80,000 were awarded to vCA's for reviewing the CA reviews.
\end{itemize}
There were $541$ CA's who submitted at least one (possibly `filtered out') review. In total, they produced: 
\begin{itemize}
    \item 357 `excellent' reviews.
    \item 4,832 `good' reviews.
    \item 3,971 `filtered Out' reviews (including many that were duplicates, algorithmically generated or empty).
\end{itemize}
On average, there were $5.5$ reviews per proposal, with total quality $6.3$.  These results suggest that the rewards for `excellent' reviews were not as high as they should be, given the CA's expertise. This is evident by the low number of `excellent' reviews compared to the `good' ones. Specifically, even if they were treated as being 3 times more valuable, they probably required more than 3 times the time to produce, motivating a choice of $a = 4$ instead, for the following funding rounds. In addition, we observe a high level of `filtered out' reviews. This is to be expected, given that participation is open and anonymous, with many agents leaving incomplete, duplicate or automatically generated reviews that were easily discarded. Despite this, fewer than 3 proposals had no reviews at all, showing that these `spam' reviews did not deter the other CA's, who reviewed every proposal knowing they would eventually be rewarded. Therefore, the \emph{coverage} guarantees were mostly attained, in line with the theoretical results.

%%%%%%%%%%%%%%%%%%%%%%%%%%%%%%%%%%%%%%%%
\section{Conclusion and Future Directions}
%%%%%%%%%%%%%%%%%%%%%%%%%%%%%%%%%%%%%%%

Our work leaves several intriguing open questions. In particular, it would be interesting to further explore both of the presented scenarios, and see whether our results can be extended when the set of the available strategies is richer. Although it might initially seem that this is the case, even answering the question of whether pure Nash equilibria exist in such cases seems challenging, as the potential function that we provide in Section \ref{sec:mprop} no longer works, and a proper modification is not trivial. In addition, for the case of multiple proposals, a natural direction would be to examine if our mechanism can also achieve some guarantees with respect to the quality objective as well, and in general, whether these two objectives are compatible-can be achieved (at some degree) at the same time. Finally, it would be interesting to see whether our mechanism is the optimal one under this setting, or if it is possible to provide better guarantees by applying different payment schemes.

\bibliographystyle{abbrvnat}
\bibliography{main}

\begin{thebibliography}{35}
\providecommand{\natexlab}[1]{#1}
\providecommand{\url}[1]{\texttt{#1}}
\expandafter\ifx\csname urlstyle\endcsname\relax
  \providecommand{\doi}[1]{doi: #1}\else
  \providecommand{\doi}{doi: \begingroup \urlstyle{rm}\Url}\fi

\bibitem[Amanatidis et~al.(2016)Amanatidis, Birmpas, and
  Markakis]{AmanatidisBM16}
G.~Amanatidis, G.~Birmpas, and E.~Markakis.
\newblock Coverage, matching, and beyond: New results on budgeted mechanism
  design.
\newblock In Y.~Cai and A.~Vetta, editors, \emph{Web and Internet Economics -
  12th International Conference, {WINE} 2016, Montreal, Canada, December 11-14,
  2016, Proceedings}, volume 10123 of \emph{Lecture Notes in Computer Science},
  pages 414--428. Springer, 2016.

\bibitem[Amanatidis et~al.(2017)Amanatidis, Birmpas, and
  Markakis]{AmanatidisBM17}
G.~Amanatidis, G.~Birmpas, and E.~Markakis.
\newblock On budget-feasible mechanism design for symmetric submodular
  objectives.
\newblock In N.~R. Devanur and P.~Lu, editors, \emph{Web and Internet Economics
  - 13th International Conference, {WINE} 2017, Bangalore, India, December
  17-20, 2017, Proceedings}, volume 10660 of \emph{Lecture Notes in Computer
  Science}, pages 1--15. Springer, 2017.

\bibitem[Anari et~al.(2014)Anari, Goel, and Nikzad]{AnariGN14}
N.~Anari, G.~Goel, and A.~Nikzad.
\newblock Mechanism design for crowdsourcing: An optimal 1-1/e competitive
  budget-feasible mechanism for large markets.
\newblock In \emph{55th {IEEE} Annual Symposium on Foundations of Computer
  Science, {FOCS} 2014, Philadelphia, PA, USA, October 18-21, 2014}, pages
  266--275. {IEEE} Computer Society, 2014.

\bibitem[Anshelevich and Postl(2016)]{AnshelevichP16}
E.~Anshelevich and J.~Postl.
\newblock Profit sharing with thresholds and non-monotone player utilities.
\newblock \emph{Theory Comput. Syst.}, 59\penalty0 (4):\penalty0 563--580,
  2016.

\bibitem[Anshelevich et~al.(2021)Anshelevich, Filos{-}Ratsikas, Shah, and
  Voudouris]{AnshelevichF0V21}
E.~Anshelevich, A.~Filos{-}Ratsikas, N.~Shah, and A.~A. Voudouris.
\newblock Distortion in social choice problems: The first 15 years and beyond.
\newblock In Z.~Zhou, editor, \emph{Proceedings of the Thirtieth International
  Joint Conference on Artificial Intelligence, {IJCAI} 2021, Virtual Event /
  Montreal, Canada, 19-27 August 2021}, pages 4294--4301. ijcai.org, 2021.

\bibitem[Archak and Sundararajan(2009)]{ArchakS09}
N.~Archak and A.~Sundararajan.
\newblock Optimal design of crowdsourcing contests.
\newblock In J.~F.~N. Jr. and W.~L. Currie, editors, \emph{Proceedings of the
  International Conference on Information Systems, {ICIS} 2009, Phoenix,
  Arizona, USA, December 15-18, 2009}, page 200. Association for Information
  Systems, 2009.

\bibitem[Arnosti and Weinberg(2019)]{arnosti2018bitcoin}
N.~Arnosti and S.~M. Weinberg.
\newblock Bitcoin: {A} natural oligopoly.
\newblock In A.~Blum, editor, \emph{10th Innovations in Theoretical Computer
  Science Conference, {ITCS} 2019, January 10-12, 2019, San Diego, California,
  {USA}}, volume 124 of \emph{LIPIcs}, pages 5:1--5:1. Schloss Dagstuhl -
  Leibniz-Zentrum f{\"{u}}r Informatik, 2019.

\bibitem[Augustine et~al.(2015)Augustine, Chen, Elkind, Fanelli, Gravin, and
  Shiryaev]{AugustineCEFGS15}
J.~Augustine, N.~Chen, E.~Elkind, A.~Fanelli, N.~Gravin, and D.~Shiryaev.
\newblock Dynamics of profit-sharing games.
\newblock \emph{Internet Math.}, 11\penalty0 (1):\penalty0 1--22, 2015.

\bibitem[Bachrach et~al.(2013)Bachrach, Syrgkanis, and Vojnovic]{BachrachSV13}
Y.~Bachrach, V.~Syrgkanis, and M.~Vojnovic.
\newblock Incentives and efficiency in uncertain collaborative environments.
\newblock In Y.~Chen and N.~Immorlica, editors, \emph{Web and Internet
  Economics - 9th International Conference, {WINE} 2013, Cambridge, MA, USA,
  December 11-14, 2013, Proceedings}, volume 8289 of \emph{Lecture Notes in
  Computer Science}, pages 26--39. Springer, 2013.

\bibitem[Benade et~al.(2021)Benade, Nath, Procaccia, and
  Shah]{benade2021preference}
G.~Benade, S.~Nath, A.~D. Procaccia, and N.~Shah.
\newblock Preference elicitation for participatory budgeting.
\newblock \emph{Management Science}, 67\penalty0 (5):\penalty0 2813--2827,
  2021.

\bibitem[Bil{\`{o}} et~al.(2019)Bil{\`{o}}, Gourv{\`{e}}s, and
  Monnot]{BiloGM19}
V.~Bil{\`{o}}, L.~Gourv{\`{e}}s, and J.~Monnot.
\newblock Project games.
\newblock In P.~Heggernes, editor, \emph{Algorithms and Complexity - 11th
  International Conference, {CIAC} 2019, Rome, Italy, May 27-29, 2019,
  Proceedings}, volume 11485 of \emph{Lecture Notes in Computer Science}, pages
  75--86. Springer, 2019.

\bibitem[Birmpas et~al.(2022)Birmpas, Kovalchuk, Lazos, and
  Oliynykov]{BirmpasKLO22}
G.~Birmpas, L.~Kovalchuk, P.~Lazos, and R.~Oliynykov.
\newblock Parallel contests for crowdsourcing reviews: Existence and quality of
  equilibria.
\newblock In \emph{Proceedings of the 4th {ACM} Conference on Advances in
  Financial Technologies, {AFT} 2022, Cambridge, MA, USA, September 19-21,
  2022}, pages 268--280. {ACM}, 2022.
\newblock \doi{10.1145/3558535.3559776}.

\bibitem[Chan et~al.(2020)Chan, Parkes, and Lakhani]{ChanPL20}
H.~Chan, D.~C. Parkes, and K.~R. Lakhani.
\newblock The price of anarchy of self-selection in tullock contests.
\newblock In A.~E.~F. Seghrouchni, G.~Sukthankar, B.~An, and N.~Yorke{-}Smith,
  editors, \emph{Proceedings of the 19th International Conference on Autonomous
  Agents and Multiagent Systems, {AAMAS} '20, Auckland, New Zealand, May 9-13,
  2020}, pages 1795--1797. International Foundation for Autonomous Agents and
  Multiagent Systems, 2020.

\bibitem[Chawla et~al.(2019)Chawla, Hartline, and Sivan]{ChawlaHS19}
S.~Chawla, J.~D. Hartline, and B.~Sivan.
\newblock Optimal crowdsourcing contests.
\newblock \emph{Games Econ. Behav.}, 113:\penalty0 80--96, 2019.

\bibitem[DiPalantino and Vojnovic(2009)]{DiPalantinoV09}
D.~DiPalantino and M.~Vojnovic.
\newblock Crowdsourcing and all-pay auctions.
\newblock In J.~Chuang, L.~Fortnow, and P.~Pu, editors, \emph{Proceedings 10th
  {ACM} Conference on Electronic Commerce (EC-2009), Stanford, California, USA,
  July 6--10, 2009}, pages 119--128. {ACM}, 2009.

\bibitem[Elkind et~al.(2021)Elkind, Ghosh, and Goldberg]{elkind2021contest}
E.~Elkind, A.~Ghosh, and P.~W. Goldberg.
\newblock Contest design with threshold objectives.
\newblock In \emph{Web and Internet Economics - 17th International Conference,
  {WINE} 2021, Potsdam, Germany, December 14-17, 2021, Proceedings}, volume
  13112 of \emph{Lecture Notes in Computer Science}, page 554. Springer, 2021.

\bibitem[Elkind et~al.(2022)Elkind, Ghosh, and
  Goldberg]{elkind2022simultaneous}
E.~Elkind, A.~Ghosh, and P.~W. Goldberg.
\newblock Simultaneous contests with equal sharing allocation of prizes:
  Computational complexity and price of anarchy.
\newblock \emph{arXiv preprint arXiv:2207.08151}, 2022.

\bibitem[Gavious and Minchuk(2014)]{GaviousM14}
A.~Gavious and Y.~Minchuk.
\newblock Revenue in contests with many participants.
\newblock \emph{Oper. Res. Lett.}, 42\penalty0 (2):\penalty0 119--122, 2014.

\bibitem[Glazer and Hassin(1988)]{GH88}
A.~Glazer and R.~Hassin.
\newblock Optimal contests.
\newblock \emph{Economic Inquiry}, 26(1):\penalty0 133--143, 1988.

\bibitem[Goel et~al.(2014)Goel, Nikzad, and Singla]{GoelNS14}
G.~Goel, A.~Nikzad, and A.~Singla.
\newblock Mechanism design for crowdsourcing markets with heterogeneous tasks.
\newblock In J.~P. Bigham and D.~C. Parkes, editors, \emph{Proceedings of the
  Seconf {AAAI} Conference on Human Computation and Crowdsourcing, {HCOMP}
  2014, November 2-4, 2014, Pittsburgh, Pennsylvania, {USA}}. {AAAI}, 2014.

\bibitem[Gollapudi et~al.(2017)Gollapudi, Kollias, Panigrahi, and
  Pliatsika]{GollapudiKPP17}
S.~Gollapudi, K.~Kollias, D.~Panigrahi, and V.~Pliatsika.
\newblock Profit sharing and efficiency in utility games.
\newblock In K.~Pruhs and C.~Sohler, editors, \emph{25th Annual European
  Symposium on Algorithms, {ESA} 2017, September 4-6, 2017, Vienna, Austria},
  volume~87 of \emph{LIPIcs}, pages 43:1--43:14. Schloss Dagstuhl -
  Leibniz-Zentrum f{\"{u}}r Informatik, 2017.

\bibitem[Gravin et~al.(2019)Gravin, Jin, Lu, and Zhang]{GravinJLZ19}
N.~Gravin, Y.~Jin, P.~Lu, and C.~Zhang.
\newblock Optimal budget-feasible mechanisms for additive valuations.
\newblock In A.~Karlin, N.~Immorlica, and R.~Johari, editors, \emph{Proceedings
  of the 2019 {ACM} Conference on Economics and Computation, {EC} 2019,
  Phoenix, AZ, USA, June 24-28, 2019}, pages 887--900. {ACM}, 2019.

\bibitem[Koutsoupias and Papadimitriou(2009)]{KP09}
E.~Koutsoupias and C.~H. Papadimitriou.
\newblock Worst-case equilibria.
\newblock \emph{Comput. Sci. Rev.}, 3\penalty0 (2):\penalty0 65--69, 2009.

\bibitem[Leonardi et~al.(2021)Leonardi, Monaco, Sankowski, and
  Zhang]{LeonardiMSZ21}
S.~Leonardi, G.~Monaco, P.~Sankowski, and Q.~Zhang.
\newblock Budget feasible mechanisms on matroids.
\newblock \emph{Algorithmica}, 83\penalty0 (5):\penalty0 1222--1237, 2021.

\bibitem[Luo et~al.(2016)Luo, Das, Tan, and Xia]{luo2016incentive}
T.~Luo, S.~K. Das, H.~P. Tan, and L.~Xia.
\newblock Incentive mechanism design for crowdsourcing: An all-pay auction
  approach.
\newblock \emph{ACM Transactions on Intelligent Systems and Technology (TIST)},
  7\penalty0 (3):\penalty0 1--26, 2016.

\bibitem[Moldovanu and Sela(2001)]{MS01}
B.~Moldovanu and A.~Sela.
\newblock The optimal allocation of prizes in contests.
\newblock \emph{American Economic Review}, 91(3):\penalty0 542--558, 2001.

\bibitem[Monderer and Shapley(1996)]{Sap1996}
D.~Monderer and L.~S. Shapley.
\newblock Potential games.
\newblock \emph{Games and Economic Behavior}, 14\penalty0 (1):\penalty0
  124--143, 1996.

\bibitem[Nakamoto(2008)]{nakamoto2008bitcoin}
S.~Nakamoto.
\newblock Bitcoin: A peer-to-peer electronic cash system.
\newblock 2008.
\newblock URL \url{https://bitcoin.org/bitcoin.pdf}.

\bibitem[Polevoy et~al.(2014)Polevoy, Trajanovski, and de~Weerdt]{PolevoyTW14}
G.~Polevoy, S.~Trajanovski, and M.~de~Weerdt.
\newblock Nash equilibria in shared effort games.
\newblock In A.~L.~C. Bazzan, M.~N. Huhns, A.~Lomuscio, and P.~Scerri, editors,
  \emph{International conference on Autonomous Agents and Multi-Agent Systems,
  {AAMAS} '14, Paris, France, May 5-9, 2014}, pages 861--868. {IFAAMAS/ACM},
  2014.

\bibitem[Siegel(2009)]{siegel2009all}
R.~Siegel.
\newblock All-pay contests.
\newblock \emph{Econometrica}, 77\penalty0 (1):\penalty0 71--92, 2009.

\bibitem[Singer(2010)]{Singer10}
Y.~Singer.
\newblock Budget feasible mechanisms.
\newblock In \emph{51th Annual {IEEE} Symposium on Foundations of Computer
  Science, {FOCS} 2010, October 23-26, 2010, Las Vegas, Nevada, {USA}}, pages
  765--774. {IEEE} Computer Society, 2010.

\bibitem[Taylor(1995)]{T95}
C.~Taylor.
\newblock Digging for golden carrots: an analysis of research tournaments.
\newblock \emph{American Economic Review}, 85(4):\penalty0 872--90, 1995.

\bibitem[Vetta(2002)]{Vetta02}
A.~Vetta.
\newblock Nash equilibria in competitive societies, with applications to
  facility location, traffic routing and auctions.
\newblock In \emph{43rd Symposium on Foundations of Computer Science {(FOCS}
  2002), 16-19 November 2002, Vancouver, BC, Canada, Proceedings}, page 416.
  {IEEE} Computer Society, 2002.

\bibitem[Vojnovic(2016)]{V2016}
M.~Vojnovic.
\newblock \emph{Contest Theory: Incentive Mechanisms and Ranking Methods}.
\newblock Cambridge University Press, 2016.

\bibitem[Zhang et~al.(2019)Zhang, Oliynykov, and Balogun]{treasury_paper}
B.~Zhang, R.~Oliynykov, and H.~Balogun.
\newblock A treasury system for cryptocurrencies: Enabling better collaborative
  intelligence.
\newblock In \emph{26th Annual Network and Distributed System Security
  Symposium, {NDSS} 2019, San Diego, California, USA, February 24-27, 2019}.
  The Internet Society, 2019.

\end{thebibliography}

\end{document}